\ifdefined\fullversion
\else
\def\fullversion{1}   % 0 = conference version; 1 = full version
\fi

\documentclass{llncs}
\usepackage{scrextend}
\changefontsizes{11pt}

\ifnum\fullversion=1	
\usepackage[a4paper, margin=1.1in]{geometry}
\else
\fi

\usepackage{amssymb}
\usepackage{algorithm}
\usepackage[noend]{algpseudocode}
\usepackage{graphicx}
\usepackage{listings}
\usepackage[perpage]{footmisc}
\usepackage[colorlinks=true,linkcolor=black,citecolor=blue]{hyperref}
\usepackage{tikz}
\usetikzlibrary{patterns,arrows.meta}
\usepackage{circuitikz}
\usepackage[braket, qm]{qcircuit}
\usepackage{stmaryrd}
\usepackage{cite}
\usepackage{array}
\usepackage{todonotes}
\usepackage{lastpage}
\usepackage{enumerate}
\usepackage{enumitem}
\usepackage{dsfont}
\usepackage{comment}
\usepackage{cryptocode}
\usepackage{multirow}
\usepackage{physics}%for ketbra notation
\usepackage{qcircuit}
\usepackage{adjustbox}
\usepackage{datetime}

\usepackage[normalem]{ulem}

\newcommand\extrafootertext[1]{%
    \bgroup
    \renewcommand\thefootnote{\fnsymbol{footnote}}%
    \renewcommand\thempfootnote{\fnsymbol{mpfootnote}}%
    \footnotetext[0]{#1}%
    \egroup
}

\let\proof\undefined
\let\endproof\undefined
\usepackage{amsthm}

\newtheorem{theoremintro}{Theorem}

\newtheorem*{keylemmaintro}{Key Lemma}

\numberwithin{theorem}{section}
\numberwithin{lemma}{section}
\numberwithin{corollary}{section}

\setlist[itemize]{label*=-}
\definecolor{amber}{rgb}{1.0, 0.49, 0.0}
\definecolor{vermillion}{RGB}{213,94,0}
\definecolor{bluishgreen}{RGB}{0,158,115}

\newcommand*{\eqdef}{:=}

\renewcommand{\vec}[1]{\mathbf{#1}}
\renewcommand{\tr}{\mathrm{Tr}}

\newcommand{\ver}{\mathsf{Ver}}

\newcommand{\eps}{\varepsilon}

\newcommand{\kgen}{\mathsf{KeyGen}}
\newcommand{\sgen}{\mathsf{StateGen}}
\newcommand{\view}{\mathrm{View}}

\newcommand{\gap}{\mathsf{gap}}

\newcommand{\adv}{\mathcal{A}}

\newcommand{\Dec}{\mathsf{Dec}}
\newcommand{\Enc}{\mathsf{Enc}}
\newcommand{\Sim}{\mathsf{Sim}}

\newcommand{\poly}{\mathsf{poly}}
\newcommand{\negl}{\mathsf{negl}}
\newcommand{\subexp}{\mathsf{subexp}}

\newcommand{\red}{R}

\newcommand{\map}{R}

\newcommand{\SPL}{\mathsf{SL}}
\newcommand{\OSPL}{1\mathsf{SL}}

\newcommand{\OSRE}{1\mathsf{RE}}

\newcommand{\QSPL}{\mathsf{QSL}}
\newcommand{\OQSPL}{1\mathsf{QSL}}

\newcommand{\QRE}{\mathsf{QRE}}
\newcommand{\OQRE}{1\mathsf{QRE}}

\newcommand{\BPP}{\mathsf{BPP}}
\newcommand{\BQP}{\mathsf{BQP}}
\newcommand{\NP}{\mathsf{NP}}
\newcommand{\UP}{\mathsf{UP}}
\newcommand{\PH}{\mathsf{PH}}

\newcommand{\coNPpoly}{\mathsf{coNP}/\mathsf{Poly}}

\newcommand{\SZK}{\mathsf{SZK}}
\newcommand{\NISZK}{\mathsf{NISZK}}
\newcommand{\NICZK}{\mathsf{NICZK}}
\newcommand{\CZK}{\mathsf{CZK}}
\newcommand{\SRE}{\mathsf{RE}}

\newcommand{\OR}{\textsc{Or}}
\newcommand{\AND}{\textsc{And}}
\newcommand{\PARITY}{\textsc{Parity}}
\newcommand{\MOD}{\textsc{Mod}}

\newcommand{\SD}{\textsc{SD}}
\newcommand{\MAJ}{\textsc{Maj}}

\newcommand{\Th}{\textsc{Threshold}}
\newcommand{\cute}{\text{sparsely lossy}}

\newcommand{\cuteness}{\text{sparse lossiness}}

\newcommand{\SAT}{3\textsc{Sat}}

\newcommand{\uSAT}{\textsc{Unique}\CircSAT}

\newcommand{\CircSAT}{\textsc{Sat}}

\newcommand{\unif}{\mathcal{U}}
\newcommand\pr[2][]{\Pr_{#1}\left[#2\right]} %use it as \pr[bottom]{}

\newcommand{\rnd}{\stackrel{{}_\$}{\leftarrow}}
\newcommand{\perm}{\mathfrak{S}}
\newcommand{\supp}{\mathrm{Supp}}

\newcommand{\fintro}{R}

\newcommand{\rem}[1]{}

\usepackage{makecell}

\usepackage{pifont} % provides \ding commands
\usepackage[most]{tcolorbox}
\tcbuselibrary{skins,breakable}
\RequirePackage{varwidth}
\usepackage{caption}

\usepackage{threeparttable}

\newtcolorbox{titlebox}[3][]{%
	enhanced,
	center,
	boxrule=0.5pt,
	arc=0.5pt,
	title={{\normalsize\sffamily\bfseries #2}~#3},
	attach boxed title to top left={xshift=+3mm,yshift*=-3mm},
	breakable=false,
	colback=white,
	colframe=black,
	colbacktitle=white,
	coltitle=black,
	boxed title style={boxrule=0.75pt, colframe=black, arc=1pt},
	#1
}

\begin{document}

\ifnum\fullversion=1	
\title{
Amplifying Randomized Encodings \& Applications\\
}
\else
\title{
Amplifying Randomized Encodings \& Applications
\fi

\author{Pouria Fallahpour\inst{1} \and Alex B. Grilo \inst{1} \and Garazi Muguruza\inst{2} ~\\ \and Mahshid Riahinia \inst{3}\thanks{Part of this work was done when the author was visiting IRIF, Paris.}}

\institute{Sorbonne Universit\'e, CNRS, LIP6, France ~\\ \email{pouria.fallahpour,alex.bredariol-grilo@lip6.fr}  \and SURF, The Netherlands ~\\\email{garazi.muguruzalasa@surf.nl} \and DIENS, \'Ecole Normale Sup\'erieure, CNRS, Paris, France ~\\ \email{mahshid.riahinia@ens.fr}}

\maketitle

\begin{abstract}
A randomized encoding for a promise problem is a randomized reduction whose outcome distribution on input $x$ can be simulated within some distance $d$, called privacy, using only one bit of information about $x$: whether it is a YES or NO instance. The encoding is one-sided if this property holds only for YES instances. Our main contribution is showing that (classical and quantum) one-sided randomized encodings have privacy and correctness amplification: any problem with a randomized encoding with privacy $1-1/\poly(n)$ and error $1/2-1/\poly(n)$, also has an encoding with negligible privacy and error. We then use this result to obtain the following statements:
\begin{itemize}[leftmargin=4mm]
	\item We show that NISZK, the class of problems that have a non-interactive (statistical) zero-knowledge proof,  has strong zero-knowledge amplification, i.e., every problem with a non-interactive zero-knowledge proof with zero-knowledge error $1-1/\poly(n)$, also has a non-interactive zero-knowledge proof with negligible zero-knowledge error. This solves a problem open since the work of Goldreich, Sahai, and Vadhan (CRYPTO '99).
	
	\item
	The worst-case hardness of a problem that has a perfect (i.e., zero error) one-sided encoding
	implies one-way functions (OWFs). And if the encoding is quantum, it implies one-way state generators (OWSGs). We then conclude that by removing the error from one-sided randomized encodings, the worst-case hardness of the complexity class SZK becomes sufficient for the existence of OWFs. 
	
	\item \emph{Weak} and \emph{imperfect} indistinguishability obfuscation (iO) implies OWFs assuming that the Polynomial Hierarchy does not collapse to its third level. Here, weakness means that the computational distance between the obfuscated circuit and the original circuit is $1-1/\poly(n)$, and imperfectness means that the error is $1/2-1/\poly(n)$.
\end{itemize}
To achieve this privacy and correctness amplification, we study randomized encodings through the lens of lossy reductions (Ball et al. ITCS 2020). 
In particular, we introduce
an extended and flexible notion of lossy reductions,
and show that it is equivalent to randomized encodings. This equivalence serves as the base of our main results and might be of independent interest.
\end{abstract}

\section{Introduction}
For a function $f$ and an input $x$, can one encode $x$ into a string $z_x$ that reveals nothing about $x$ beyond $f(x)$? Of course, the value $f(x)$ is itself a valid candidate for the encoding, so the question becomes non-trivial if we require that computing the encoding be simpler than computing $f$. It can be shown that deterministic encodings cannot generally satisfy this requirement~\cite{Apple11}. However, randomness changes the picture: for a wide range of complex functions $f$, there exist randomized encodings that are far simpler to evaluate than $f$ itself. The gap between the complexity of $f(x)$ and that of $x \mapsto z_x$, together with the privacy of $z_x$, makes randomized encodings a rich tool for complexity theory as well as cryptography.

Ishai and Kushilevitz~\cite{IK00} showed that every function computable by a polynomial-size non-deterministic branching program ($\mathsf{NL}/\poly$) has a randomized encoding that can be evaluated as a polynomial of degree~$3$. Consequently, securely evaluating any such function reduces to securely evaluating degree-$3$ polynomials, which yields round-efficient secure-computation protocols for all $\mathsf{NL}/\poly$. Applebaum, Ishai, and Kushilevitz~\cite{AIK06} pushed this further, and replaced the degree-$3$ polynomials by constant-depth circuits of locality~$4$ (denoted by $\mathsf{NC}^0_4$). The evaluation of any function in $\mathsf{NL}/\poly$, including candidate pseudorandom generators (PRGs) and pseudorandom functions (PRFs), can thus be brought all the way down to $\mathsf{NC}^0_4$, which provides an evidence to the utility of randomized encodings.
Brakerski and Yuen~\cite{BY22} introduced a quantum analogue, i.e., quantum encodings, and showed applications in quantum zero-knowledge proofs.

More generally, for any promise problem $\Pi$ one can study randomized encodings of its characteristic function~$\chi_\Pi$, which on input $x$ outputs $1$ if $x \in \Pi_Y$ and $0$ if $x \in \Pi_N$. In this terms, the discussion above implies that the class of problems with $\mathsf{NC}^0_4$ encodings contains $\mathsf{NL}/\poly$. What if we ask only that the encoding be computable in polynomial time? The resulting class has already been characterized~\cite{AR16}:
$\SRE[d(n),\mu(n)]$ denotes the class of all problems~$\Pi$ for which there exists a polynomial-time algorithm~$\Enc$ (the randomized encoding) that satisfies the following two properties:
\begin{itemize}[leftmargin=3mm]
\item ($d$-Privacy) There exists a polynomial-time algorithm $\Sim$ such that for all $n$-bit valid instances~$x$:
$$\|\Sim(1^n, \chi_\Pi(x)) - \Enc(x)\|\leq d(n)~,$$
where $\|X - Y \|$ denotes the statistical distance between two distributions $X$ and $Y$.

\item ($\mu$-Correctness) There exists a (possibly unbounded) algorithm $\Dec$ such that for all $n$-bit valid instances~$x$:
	$$\Pr \left[ \Dec(\Enc(x)) \neq \chi_\Pi (x) \right] \leq \mu(n)~.$$ 
\end{itemize}
If the existence of the simulator is only guaranteed for the YES instances, the encoding is called \emph{one-sided} and the corresponding class of promise problems is denoted by~$\OSRE[d(n),\mu(n)]$. When $d(n),\mu(n) = \negl(n)$, we drop these arguments for simplicity.

The class $\SRE$ trivially contains $\BPP$, that is the class of problems admitting probabilistic polynomial-time solvers with bounded error. However, no explicit non-trivial encoding is known for any problem beyond $\BPP$. There is evidence that the containment is nonetheless strict: the work of~\cite{AIKP15} showed that under suitable computational hardness assumptions, $\SRE$ must be strictly larger than $\BPP$. On the flip side, Applebaum and Raykov~\cite{AR16} showed that $\OSRE$ is equivalent to  $\NISZK$, the class of all problems admitting a non-interactive (statistical) zero-knowledge proof system. Since $\SRE \subseteq \OSRE$, this bounds both classes. Despite these results, $\SRE$ and $\OSRE$ remain little understood, and yet important objects of study. Randomized encodings in constant depth gave us cryptography of remarkably low complexity; in the polynomial-time regime, separating $\SRE$ from $\BPP$ would give us cryptography of another kind, namely OWFs~\cite{AR16}. Moreover, the close connection between polynomial-time randomized encodings and non-interactive zero-knowledge proofs makes them a natural tool for studying proof systems.
In this work, we study the properties of randomized encodings and show various applications to complexity theory and cryptography.

\subsection{Our Contribution} \label{sec:contrib}
{
In this work, we introduce a new technique for \emph{privacy and correctness amplification} of $1\SRE$.
More precisely, we prove the following lemma:
\begin{keylemmaintro}[Privacy \& Correctness Amplification of 1$\SRE$]\label{intro:priv-amp-thm}
	For any sufficiently large $n \in \mathbb{N}$, it holds that 
	$$\OSRE \left[ 1-1/\poly(n), 1/2-1/\poly(n) \right] \subseteq \OSRE\left[\negl(n), \negl(n) \right]\, ,$$ 
	up to some polynomial-size non-uniform advice. Moreover, if the error of the left-hand side is zero, the error of the right-hand side will also be zero. 
\end{keylemmaintro}

The above result is obtained by studying the \emph{lossy} properties of randomized encodings. More precisely, we show that the class $1\SRE$ can be equivalently characterized by a quantitative measure that we call \emph{sparse lossiness}.
Let us first define this measure formally. 
Our starting point is the notion of \emph{lossy reductions}, introduced by Ball et al.~\cite{BBDD+20}. 
In their definition, a reduction $R$ for a promise problem $\Pi$ is $\lambda(n)$-lossy if $I(X;R(X)) \leq \lambda(n)$ for \emph{all} distributions $X$ supported on $\{0,1\}^n$. We relax this requirement as follows. Call $X$ a \emph{sparse uniform} distribution if it is uniform over a set of $\poly(n)$ strings that lies entirely in $\Pi_Y \cap \{0,1\}^n$ or entirely in $\Pi_N \cap \{0,1\}^n$. We then say that a reduction $R$ is $\lambda(n)$-\emph{$\cute$} if $I(X;R(X)) \leq \lambda(n)$ for all such $X$. In other words, we only require $R$ to lose information on sparse uniform distributions of YES or NO instances.
We let $\SPL[\lambda(n),\mu(n)]$ denote the class of all problems that have $\lambda(n)$-$\cute$ reductions with correctness error $\mu(n)$. Moreover, the reduction is said to be \emph{one-sided $\cute$} if it is $\cute$ only on sparse uniform distributions over YES instances; we denote the corresponding class of problems by $\OSPL[\lambda(n),\mu(n)]$. This refined definition allows us to prove the following theorem:

\begin{theoremintro}\label{intro:sl-re-th}
For any error function~$\mu(n) \leq 1/2 - 1/\poly(n)$ the following statements hold:
 \begin{enumerate}
\item One-sided randomized encodings with inverse-polynomial privacy are equivalent to one-sided sparsely lossy reductions:
\begin{itemize}
\item $\forall d(n) \leq 1-1/\poly(n)$, we have $\OSRE\left[d(n),\mu(n)\right] \subseteq \OSPL[\negl(n),\negl(n)]$. Moreover, if $\mu(n)=0$, the error of the right-hand side is $0$.
{
\item We have $\OSPL[O(\log n),\mu(n)] \subseteq \OSRE\left[\negl(n),\mu(n)\right] $. 
}
\end{itemize}
\vspace{.2cm}
 \item Randomized encodings are partially equivalent to sparsely lossy reductions:
 \begin{itemize}
 \item We have $\SRE[\negl(n),\mu(n)] \subseteq \SPL[\negl(n) ,\mu(n)] $.
 \item We have $\SPL[O(\log n),\mu(n)]  \subseteq \SRE\left[1-1/\poly(n),\mu(n)\right]$.\footnote{ The inclusion can be extended to $\cute$ $f$-reductions, the reductions that take multiple instance as input such as an ~$\OR$ reduction. More details can be found in Section~\ref{sec:spl}.}
 \end{itemize}
 \end{enumerate}
 The statements hold up to some polynomial-size non-uniform advice.
\end{theoremintro}

One can verify that the first item of this theorem implies our Key Lemma. The conceptual connection between randomized encodings and lossy reductions was already observed by~\cite{BBDD+20}. In Theorem~\ref{intro:sl-re-th}, we make this connection precise and give an exact quantitative relation between the two.

\medskip
\noindent\textbf{Quantum Setting.} The techniques we use to prove Theorem~\ref{intro:sl-re-th} are generic, and largely independent of whether the underlying objects are quantum or classical. This allows us to obtain a quantum version of the statement. More precisely, let $\QRE$ and $\OQRE$ be the quantum analogues of $\SRE$ and $1\SRE$, where the encoding and the simulator are allowed to be quantum algorithms. The inputs, which are instances of computational problems, remain classical, and statistical distance is replaced by trace distance. With these definitions, Theorem~\ref{intro:sl-re-th} and Key Lemma hold verbatim for $\QRE$ and $\OQRE$.
}

{
\medskip
Despite its simple statement, the Key Lemma is a powerful tool for obtaining further results. In particular, we use it to obtain the following applications.

\subsubsection{Application 1: Zero-knowledge Amplification of NISZK.}
Recall that $\SZK$ is the class of problems who have a zero-knowledge interactive proof system. More precisely, in a zero-knowledge proof system for $\Pi$, an unbounded prover interacts with a polynomially-bounded verifier in order to prove that a given input $x \in \Pi_Y$ is a YES instance, without revealing any extra information about $x$ beyond this membership. Three quantities $(c(n),s(n),\varepsilon(n))$ measure the quality of the proof system:
\begin{itemize}
\item \textbf{Completeness}: For any~$x\in \Pi_Y \cap \{0,1\}^n$, the verifier accepts with probability at least $1-c(n)$. 
\item \textbf{Soundness}: For any~$x\in \Pi_N \cap \{0,1\}^n$ and every possibly malicious prover, the verifier accepts with probability at most $s(n)$. 
\item \textbf{Zero-Knowledge}: There exists a polynomial-time algorithm~$\Sim$ such that for every~$x\in \Pi_Y \cap \{0,1\}^n$, the distribution of $\Sim(x)$ is within $\varepsilon(n)$ statistical distance from the honest transcript of the protocol.
\end{itemize}

In its non-interactive variant, denoted $\NISZK$, the interaction between the prover and the verifier is limited to a single message, and both parties have access to a shared randomness that is sampled from an \emph{a priori} fixed distribution, known as the common reference string.

In the usual definition of $\SZK$ (respectively $\NISZK$), all three parameters are set to be negligible. Moreover, as long as $\varepsilon(n)$ is negligible, any protocol with $c(n)+s(n) \leq 1-1/\poly(n)$ can be amplified into one where all three parameters are negligible, and in fact exponentially small~\cite{GSV99,SV03}. Amplifying the zero-knowledge property itself, meaning reducing $\varepsilon(n)$ when it is not already negligible, appears to be more challenging. The one known result in this direction is due to Goldreich, Sahai, and Vadhan~\cite{GSV99}, who considered a \emph{weak} zero-knowledge property for $\NISZK$. This property requires that for every inverse polynomial $1/p(n)$, there exists a polynomial-time simulator whose output is within statistical distance $1/p(n)$ of the honest transcript. They showed that any such protocol can be amplified to one with a single simulator of negligible error. A similar result holds for $\SZK$~\cite{SV03}.
Yet, the following question has remained unanswered ever since:
\begin{center}
\emph{Question 1. Can the zero-knowledge guarantee of $\SZK$ or $\NISZK$ protocols be amplified from inverse-polynomial to negligible error?}
\end{center}

We prove the following theorem which answers this question in the affirmative for $\NISZK$:
\begin{theoremintro}\label{intro:thm-niszk}
For every $c(n),s(n),\varepsilon(n)$ satisfying $2s(n)+c(n)+\varepsilon(n) = 1-1/\poly(n)$, it holds that $\NISZK[c(n),s(n),\varepsilon(n)] \subseteq \NISZK$ (up to some non-uniform polynomial-size advice).
\end{theoremintro}

In particular, when $c(n)=s(n)=\negl(n)$ and $\varepsilon(n)=1-1/\poly(n)$, the theorem above yields a zero-knowledge amplification for $\NISZK$ from inverse-polynomial to negligible.
}

\subsubsection{Application 2: OWFs and the Worst-Case Hardness of SZK.}
One-way functions (OWFs) are fundamental cryptographic tools whose existence can be viewed as
the minimal {assumption} required for cryptography. 
It is not known whether OWFs unconditionally exist or not.  
Instead, their existence is often based on the worst-case hardness of concrete computational problems such as DLOG~\cite{DH76}, RSA~\cite{RSA78}, lattice and code-based
problems~\cite{Ajtai98,Ale03,MR07,Reg09} and more.
While a polynomial-time algorithm for any of these problems would have considerable implications for cryptography, the widespread belief of their worst-case hardness is based less on \emph{structural} or \emph{complexity-theoretic} evidence and more on the failed \emph{algorithmic} attempts in directly solving them. 
In the broader sense, relying on just a few specific problems remains limiting. 
A more robust approach is to look for OWFs based on the worst-case hardness of a large  class of problems.  
These types of connections provide stronger evidence for the existence of OWFs. 

For OWFs to exist, it is necessary that~$\mathsf{P} \neq \mathsf{NP}$. Whether this is sufficient is arguably one of the most important questions in the {theory of} cryptography. 
In a series of works~\cite{FF93,BT06,AGGM06,BB15}, it has been shown that there are black-box barriers to build OWFs from worst-case hardness of~$\NP$. In fact, in the non-uniform setting where the algorithms have access to advice, some of these barriers also hold even for $\SZK$-complete problems. Recall that $\SZK$ is the class of problems that have statistical zero-knowledge interactive proofs and is inside~$\NP \cap \mathsf{co}\NP$ (in the non-uniform setting). On the other hand, many of the widely used assumptions such as lattice-based or group-based ones~\cite{DH76,RSA78,Ajtai98,MR07,Reg09} belong to~$\SZK$.\footnote{Yet, none of them are complete for~$\SZK$.} This suggests that $\SZK$ may provide sufficient hardness for building OWFs. 
Ostrovsky~\cite{Ost91} showed that the the
average-case hardness of $\SZK$ yields OWFs. 
Later, Ostrovsky and Wigderson~\cite{OW93} could show that the worst-case hardness of $\SZK$ implies the (seemingly weaker) auxiliary-input OWFs.\footnote{ This is different from OWFs with non-uniform advice. We refer the reader to the cited paper for more details.}
A positive result relating OWFs to the worst-case hardness of $\SZK$ would therefore sheds more light on the sufficient hardness for OWFs.
\begin{center}
\emph{Question 2. Can we show that the worst-case hardness of $\SZK$ suffices for the existence of OWFs?}
\end{center}

We show that an error removal for $1\SRE$, even with loose privacy, is sufficient to prove that the worst-case hardness of $\SZK$ implies the existence of OWFs. More formally,

\begin{theoremintro}\label{intro:error-rem-thm}
	If $\OSRE[\negl(n),\negl(n)] \subseteq \OSRE\left[1-1/{\poly(n)},0\right]$ (up to some non-uniform polynomial-size advice), and $\SZK$ is worst-case hard, then OWFs exist.\footnote{We note that by the Key Lemma, the privacy can always shrink to negligible, but we keep $1-1/\poly(n)$ privacy in the statement for expressiveness.}
\end{theoremintro}

We expect that an error removal algorithm for $1\SRE$ exists.
Without considering the privacy, a polynomial-time randomized reduction can be derandomized via the textbook repeat-then-majority procedures. Here, the challenge is to manipulate the randomness in order to not completely lose the structure (privacy).%\Mnote{this feels incomplete.}

\subsubsection{Application 3: OWFs from Weak and Imperfect iO.} An Indistinguishability Obfuscation (iO) is an algorithm that takes as input a circuit $C:\{0,1\}^n\rightarrow \{0,1\}^*$ and outputs an obfuscated circuit such that (I) the latter has the same input-output behavior as the original circuit except with probability $\varepsilon=\negl(n)$, and (II) for every two functionally equivalent circuits $C,C'$, their obfuscated variants are computationally distinguishable with probability at most $\tau=\negl(n)$. Numerous works are dedicated to study 
the existence of obfuscation, e.g.,~\cite{GGH+16,BBKK18,BHJ+19,JLS21,WW21,JLS22}, and its importance in cryptography as well as complexity theory, e.g.,~\cite{GGG+14,BFM14,BPR15,CLP15,GPOS16,BGJ+16,CHJ+16,SW21}. In particular, Goldwasser and Rothblum~\cite{GR07} implicitly show that when $\tau \approx 1/10$ and $\varepsilon = \negl(n)$, iO implies OWFs unless the Polynomial Hierarchy collapses.
Komargodski et al.~\cite{KMN+14} show that iO with $\tau=\negl(n)$ and $\varepsilon = 1/2-1/\poly(n)$ implies OWFs if~$\NP \not\subseteq \BPP$. A simple argument aligned with that of~\cite{MSY25} also yields OWFs from iO when $\tau=1-1/\poly(n)$ and $\varepsilon=0$, if $\NP \not\subseteq \BPP$. All these regimes relax one of $\tau$ or $\varepsilon$ while keeping the other close to optimal. This raises the following question:
\begin{center}
\emph{Question 3. Does iO imply OWFs when both $\tau$ and $\varepsilon$ are simultaneously far from optimal?}
\end{center}
%The main challenge is to go beyond the two aforementioned regimes. 
%We establish the following relation between iO and OWF:
In this regard, we prove the following theorem:
\begin{theoremintro}\label{intro:io}
	Assuming the existence of iO with parameters $\tau(n)=1-1/\poly(n)$ and $\varepsilon(n)=1/2-1/\poly(n)$ and that the Polynomial Hierarchy does not collapse to its third level, OWFs exist.
	
	Moreover, assuming the existence of iO with subexponential runtime and parameters $\tau(n)=1-1/\subexp(n)$ and $\varepsilon(n)=1/2-1/\subexp(n)$, and that refuting $\SAT$ is exponentially hard, OWFs exist.~\footnote{Note that iO with subexponential runtime should not be confused with subexponentially-secure iO.}
\end{theoremintro}

The assumption that refuting $\SAT$ is exponentially hard is known as non-uniform non-deterministic exponential time hypothesis (nuNETH). It is introduced by~\cite{CGI+16}, and more formally states that $\SAT$ does not have a non-uniform reduction to~$\mathsf{co}\NP$ that runs in time subexponential in~$n$, where~$n$ is the number of variables in the~$\SAT$ instance. It has been further studied by~\cite{CGI+16,CRTY23}.

\begin{figure}[ht]
  \centering
  %% axis units (kept as macros so lengths can be matched across the two axes)
\def\xunit{8.5}   % cm per unit on the tau (horizontal) axis
\def\yunit{11}    % cm per unit on the epsilon (vertical) axis
\def\neglcm{0.40} % how wide "negligible" should LOOK, in cm, on both axes

\begin{tikzpicture}[
    x=\xunit cm, y=\yunit cm,
    >={Stealth[length=6pt]},
    callout/.style={->, thick, shorten >=3pt},
    lbl/.style={font=\small, align=center},
    scale=0.8
  ]

  %% ---- adjustable positions (not to scale on purpose) -----------------
  %% both negl(n) marks are \neglcm cm from the origin, on either axis
  \pgfmathsetmacro{\negl}{\neglcm/\xunit}      % negl(n) on the tau axis
  \pgfmathsetmacro{\negltwo}{\neglcm/\yunit}   % negl(n) on the epsilon axis
  %% tau  (horizontal)
  \def\tenth{0.16}     % 1/10                 on the tau axis
  \def\thi{0.92}       % 1 - 1/poly(n)        on the tau axis
  %% epsilon (vertical)
  \def\ehi{0.44}       % 1/2 - 1/poly(n)      on the epsilon axis

  %% ---- gray (dotted) region -------------------------------------------
  \fill[blue!8]                                        (\negl,0) rectangle (\thi,\ehi);
  \path[pattern=dots, pattern color=blue!45]           (\negl,0) rectangle (\thi,\ehi);
  \draw[blue!45, thick]                                (\negl,0) rectangle (\thi,\ehi);

  %% ---- hatched regions -------------------------------------------------
  \fill[black!10]                                            (0,0) rectangle (\tenth,\negltwo);
  \fill[green!10]                                           (0,0) rectangle (\negl,\ehi);
  \path[pattern=north east lines, pattern color=black!60]    (0,0) rectangle (\tenth,\negltwo);
  \path[pattern=north east lines, pattern color=green!60]   (0,0) rectangle (\negl,\ehi);
  \draw[black!60!black, thick]                               (0,0) rectangle (\tenth,\negltwo);
  \draw[green!60!black, thick]                              (0,0) rectangle (\negl,\ehi);

  %% ---- axes ------------------------------------------------------------
  \draw[->, thick] (-0.02,0) -- (1.08,0) node[right] {$\tau$};
  \draw[->, thick] (0,-0.015) -- (0,0.70) node[above] {$\varepsilon$};
  \node[below left, xshift=-6pt] at (0,0) {$0$};

  %% ---- bold red horizontal line ----------------------------------------
  \draw[red, line width=2.2pt] (0,0) -- (\thi,0);

  %% ---- ticks: tau axis (horizontal) ------------------------------------
  \foreach \x/\lab/\shift in {%
      \negl/{$\mathsf{negl}(n)$}/0cm,
      \tenth/{$\tfrac{1}{10}$}/0cm,
      0.5/{$\tfrac12$}/0cm,
      \thi/{$1-\tfrac{1}{\mathsf{poly}(n)}$}/0cm,
      1/{$1$}/-0.52cm}
    {\draw[thick] (\x,0) -- (\x,-0.009);
     \node[below, font=\small, yshift=\shift] at (\x,-0.009) {\lab};}

  %% ---- ticks: epsilon axis (vertical) ----------------------------------
  \foreach \y/\lab in {%
      \negltwo/{$\mathsf{negl}(n)$},
      \ehi/{$\tfrac12-\tfrac{1}{\mathsf{poly}(n)}$},
      0.55/{$\tfrac12$}}
    {\draw[thick] (0,\y) -- (-0.010,\y);
     \node[left, font=\small] at (-0.012,\y) {\lab};}

  %% ---- callouts (transposed too) ---------------------------------------
%  \draw[callout, red!70!black] (1.00,0.20) -- (0.88,0.006);
%  \node[lbl, right, red!70!black] at (1.01,0.20) {Follows from \cite{MSY25}};
%
%  \draw[callout, green!45!black] (0.22,0.64) -- (0.05,0.46);
%  \node[lbl, right, green!45!black] at (0.23,0.64) {By \cite{KMN+14}};
%
%  \draw[callout, blue!60!black] (0.16,-0.02) -- (0.30,-0.14);
%  \node[lbl, right, blue!60!black] at (0.31,-0.14) {By \cite{GR07}};

\end{tikzpicture}
  \caption{The set of $(\varepsilon,\tau)$-iO implying the existence of OWFs assuming that PH does not collapse. The grey and the green hatched regions are respectively the parameter sets of \cite{GR07} and \cite{KMN+14}, and the bold red line corresponds to technique of \cite{MSY25}. The dotted blue area correspond to our result.}
  \label{fig:eps-tau}
\end{figure}

\paragraph{On the existence of iO.} {In fact, we prove a stronger result concerning the existence of iO. We show that unless the PH collapses to its third level, any obfuscation algorithm with error $1/2-1/\poly(n)$ must have a \emph{statistical} distance negligibly close to 1 on at least one family of functionally-equivalent pair of circuits. 
Similarly, assuming nuNETH, any obfuscation algorithm with error $1/2-1/\subexp(n)$ must have a statistical distance exponentially close to 1.

\subsubsection{Application 4: One-Way State Generators.} We also obtain a weaker variant of Theorem~\ref{intro:error-rem-thm} for \emph{inefficient-verifier one-way state generators} (IV-OWSG). Roughly speaking, an IV-OWSG consists of a polynomial-time quantum algorithm $\sgen$ that maps $k \in \{0,1\}^n$ to a quantum state $\rho_k$ and a possibly unbounded algorithm $\ver$ that, given $k$ and a state $\rho$, either accepts or rejects. For instance, one way to instantiate $\ver$ is to check whether $\rho$ is close to $\rho_k$ in trace distance. Moreover, it must be hard for any polynomial-time adversary to find a $k'$, given multiple copies of $\rho_k$, such that $\ver(k',\rho_k)$ accepts. When $\rho_k$ is a pure quantum state for every fixed $k$, we say that the IV-OWSG is pure. Almost immediately by recycling the techniques developed for Theorem~\ref{intro:error-rem-thm}, we show that if $\OQRE_{\mathsf{pure}}\left[1-1/\poly(n),0\right]$ is (quantum) worst-case hard, then pure IV-OWSGs exist. Here $\OQRE_{\mathsf{pure}}$ refers to those quantum encodings whose $\Enc(\cdot;r)$ function returns a pure state for any fixed
possible classical randomness $r$.

\subsection{Technical Overview}\label{sec:tech-overview}

In this section, we briefly present the core technical tools that we use. We first discuss the proof of Theorem~\ref{intro:sl-re-th}, the main technique to obtain the other results. 
We then sketch the proofs of Theorems~\ref{intro:thm-niszk},~\ref{intro:error-rem-thm} and~\ref{intro:io}. Almost all our techniques immediately carry over to the quantum setting, leading to the quantum results. To avoid the risk of being redundant, we have not included the overview of these proofs in this section.

\vspace{1mm}
\noindent
\textbf{Proof sketch of Theorem~\ref{intro:sl-re-th} (Randomized encodings and $\cute$ reductions).}
In what follows, we sketch the proof of Theorem~\ref{intro:sl-re-th} separately for the two directions of the inclusions. 

\smallskip\noindent
\textit{From sparse lossiness to encodings.} 
The difficulty in going from~$\SPL$ to~$\SRE$ lies in extracting a privacy simulator out of the $\cuteness$ property. 
To make this happen, we resort to the \emph{disguising lemma} by Drucker~\cite{Dru15}.
The proof works for general \emph{$f$-reductions}, but for simplicity, we restrict the discussion to regular reductions. Let~$\Pi\in \SPL$ and~$\red$ be a reduction from $\Pi$ to some other problem that is $\cute$. The original variant of the disguising lemma %by Drucker~\cite{Dru15} 
states that if a mapping~$\red$ is compressing, then for any subset $S \subseteq \{0,1\}^n$, there exists a sparse distribution $D_S$ supported on $S$, such that for all $y \in S$, the distribution of $\fintro(y)$ is statistically close to that of $\fintro(D_S)$, where the closeness depends on the amount of compression achieved by $\fintro$---therefore, if $\fintro$ is sufficiently compressing, then $\fintro$ \emph{disguises} $y$. More precisely, $\mathbb{E}_{D_S}[\|\fintro(y)-\fintro(D_S)\|] \leq \delta^*$, where~$\delta^* \approx 1+ 2^{-\lambda} $, if~$\fintro$ compresses an $n$-bit instance to~$\lambda$ bits. Looking at it differently, since~$\red$ is a reduction from a problem~$\Pi$, we can set~$S:=\Pi_Y \cap \{0,1\}^n$ and conclude that~$\fintro(D_S)$ is close to~$\fintro(y)$ in average, for every YES instance~$y$. In other words, $\fintro(D_S)$ is a simulator for YES instances. Therefore, we obtain randomized encodings from compressing reductions. {In other words, the class of problems with compressing reductions are a subset of those with randomized encodings. }

While the original disguising lemma is stated for compressing reductions, the work of Ball et al.~\cite{BBDD+20} extends it for lossy reductions
. Our challenge is to extend it even further for $\cute$ reductions.

In order to sketch our improvements of this lemma, 
we first briefly go over Drucker's proof. Consider the following simultaneous-move two-player game: Given a set $S$ and $d \leq |S|$, Player~$1$ chooses a $d$-sparse distribution~$D$ supported on~$S$, and Player~$2$ chooses an element~$y \in S$
. Define the payoff as~$\|\fintro(y) - \fintro(D)\|$. Drucker first shows that if~$\fintro$ is compressing, then for any strategy~$\mathcal{Y}$ for choosing~$y$, there exist a $d$-sparse distribution~$D$, that is more precisely the uniform distribution over $d$ samples of $\mathcal{Y}$, such that: 
\begin{equation}\label{eq:intro-1}
\mathbb{E}_{\mathcal{Y}^{\otimes d}}[\|\fintro(y)-\fintro(D)\|] \leq \delta^* \, .
\end{equation} 
\rem{The proof then proceeds by swapping the quantifiers of the above statement using the minimax theorem; more precisely, consider a simultaneous-move two-player game where one player chooses a distribution~$D$ (subject to be uniform over some multiset of size~$d$) and the other player chooses an element~$y$, 
and let the payoff be~$\|\fintro(y) - \fintro(D)\|$. For any strategy~$\mathcal{Y}$ for choosing~$y$, let~$D$ be the uniform distribution over $d$ samples of $\mathcal{Y}$.
Then, Equation~\eqref{eq:intro-1} bounds the expected payoff from above.} By the minimax theorem, we can switch the quantifiers in above, implying that there exists a distribution~$\mathcal{D}_S$ 
such that, for any choice of~$y \in S$, it holds that~$\mathbb{E}_{D\sim \mathcal{D}_S}[\|\fintro(y) - \fintro(D)\|] \leq \delta^*$. However, note that~$\mathcal{D}_S$ is not guaranteed to be sparse. Drucker finally uses a result by Lipton and Young~\cite[Theorem~2]{LY02} which roughly states that restricting the strategies of Player~$1$ to uniform strategies with support size~$\ln (\#\{\text{choices of Player }2\}) /\gamma^2$ only changes the optimal expectation payoff with an additive factor~$\gamma$. \footnote{The same holds for Player~$2$.} This therefore \emph{sparsifies} the support of~$\mathcal{D}_S$
 . Conversely, it incurs a loss of at most an additive factor~$\gamma$ in the expectation bound in Equation~\eqref{eq:intro-1} and obtains~$\delta^* + \gamma$.

We relax the requirement on~$\fintro$ and show that \emph{$\cuteness$} of~$\fintro$ suffices to obtain a similar result. More precisely, we show that if~$\fintro$ loses information on input distributions that are uniform with support size roughly~$1/\gamma^3$, then one loses nothing but an(other) additive factor~$\gamma$ in Equation~\eqref{eq:intro-1}. This relies on a double use of the result by Lipton and Young~\cite[Theorem~2]{LY02}; we apply it once for Player~2 and once more for Player~1. 
Before using the minimax theorem, we restrict Player~$2$'s strategies to be uniform distributions with support size roughly~$1 /\gamma^3$, and we choose~$d \approx 1/\gamma$. By showing that Equation~\eqref{eq:intro-1} remains correct even with this new restriction, we obtain an additive~$\gamma$-approximation of the value of the game (first use of~\cite[Theorem~2]{LY02} for Player~2). Following the minimax theorem, and sparsifying~$\mathcal{D}_S$ (via a second use of~\cite[Theorem~2]{LY02} for Player~1), we achieve the final upper bound~$\delta^*+2\gamma$. 
Note the constraint that~$y$ must have the same support as the distribution that it replaces. 
Our variant of disguising lemma with $\cute$ reductions also extends to this setting (we refer to Appendix~\ref{app:disguising-proof} for more details),
{
and allows us to prove 
\begin{align}\label{eq:intro-log}
\SPL[O(\log n),\mu(n)] &\subseteq \SRE[1-1/\poly(n),\mu(n)]/\poly(n) \, , \\
\OSPL[O(\log n),\mu(n)] &\subseteq \OSRE[1-1/\poly(n),\mu(n)]/\poly(n)\, . \nonumber
\end{align}
For the one-sided variant, we show how to reduce the privacy to negligible in the next parts.
}

\smallskip\noindent
\textit{From encodings to sparse lossiness.}  
Assume that $\Pi \in \SRE$. Recall that the privacy property of the randomized encoding puts an upper-bound on the statistical distance between the output distributions of the encoding~$\Enc(x)$ and the simulator~$\Sim(\chi_\Pi(x))$. Since $\Sim$ depends only on~$\chi_\Pi(x)$ (i.e., whether $x$ is a YES or NO instance), rather than the whole input $x$, intuitively, {$\Enc$ should discard all information about its input beyond its membership in the language}. Yet, the proof is not immediate. We sketch this proof below. 

We first explain why~$\Enc$ is sparsely lossy for the case where privacy~$d$ is negligible.
We have to show that~$I(X;\Enc(X))$ is bounded for every distribution $X$ that is sparse, splitting and uniform, as required by the definition of the sparse lossiness, provided earlier in this section. For any two random variables~$X$ and~$Y$, it holds that $I(X,Y) = \mathbb{E}_Y[D_{\textrm{KL}}(p_{X|Y} \| p_X)]$ where~$D_{\textrm{KL}}$ is the KL-divergence. 
It is therefore sufficient to bound~$D_{\textrm{KL}}\left( p_{X | \Enc(X)} \ \| \ p_{X} \right)$.
We do this by leveraging the \emph{reverse Pinsker's inequality}; for any two distributions~$p \ll q$ (i.e., $p$ is absolutely continuous with respect to $q$), the reverse Pinsker's inequality bounds~$D_{\textrm{KL}}(p\|q) \leq (\log e/\alpha_q) \cdot \Delta(p,q)^2$, where~$\alpha_q = \min_x \{q(x) > 0\}$. {Applying this bound and taking the expectation over $\Enc(X)$ yields an upper bound on the lossiness of $\Enc$ as follows: $I(X;\Enc(X)) \leq O(d^2/\alpha_{X})$, where~$d$ denotes the privacy of the encoding.} {Since~$X$ is uniform over $\poly(n)$-many elements, it holds that $\alpha_X = 1/\poly(n)$. Therefore, if~$d=\negl(n)$, we have that that~$I(X;\Enc(X))\leq \negl(n)$. In other words, $\Enc$ is a sparsely lossy reduction for~$\Pi$ as long as~$d=\negl(n)$.}

For extending this result to larger values of~$d(n)$, we use the \emph{XOR lemma}~\cite{SV03}. In a nutshell, XOR lemma {takes as input a polynomial $q(n)$} and transforms a pair of circuits into a new pair whose statistical distance is that of the initial pair to the power of~$q(n)$. Suppose that $d(n)=1-1/\poly(n)$
. Applying it to~$(\Enc,\Sim)$ yields a new pair~$(P,T)$. 
A similar argument as in the previous paragraph gives us an upper bound on~$I(X;P(X)) \leq O(d^{2q}/\alpha_{X}) = \negl(n)$, when $q(n)$ is a sufficiently large polynomial. Therefore, it follows that $P$ is a sparsely lossy {mapping} for $\Pi$. However the correctness {error} of $P$ is not necessarily negligible. In fact, 
{one can show that $\Delta(P(x_Y),P(x_N)) \geq 1-2\mu$ for YES and NO instances~$x_Y$ and~$x_N$, when $\Delta(\Enc(x_Y),\Enc(x_N)) \geq 1-2\mu$. In other words, the error remains the same (see Lemma~\ref{lemma:xor-trick} for details).} In order to reduce the error, we consider the new amplified reduction $P^{\otimes \ell}$, where we repeat $P$ for $\ell$ times. By the direct product lemma~\cite{SV03}, for a sufficiently large polynomial $\ell(n)$, we get $\Delta(P^{\otimes \ell}(x_Y),P^{\otimes \ell}(x_N)) \geq 1-2\exp(-\ell (1-2\mu)^2/2) = 1-\negl(n)$ while the effect on the mutual information is $I(X;P^{\otimes \ell}(X)) \leq O(\ell \cdot d^{2q}/\alpha_{X}) = \negl(n)$.  {This gives us
\begin{equation}\label{eq:intro-re-to-sl}
\OSRE[1-1/\poly(n),\mu(n)] \subseteq \OSPL[\negl(n),\negl(n)] \, .
\end{equation}
}

For { two-sided} randomized encodings, i.e., when the simulator works on both YES and NO instances, the amplification no longer works.
{In the output of the XOR lemma, $P$ (and similarly $T$) is not solely obtained from $\Enc$ and internally runs~$\Sim$ as well.} Let~$\Sim(1)$ and~$\Sim(0)$ be the simulators for YES and NO instances, respectively. The issue comes from the fact that we obtain different $P$ circuits when we apply the XOR lemma to $(\Enc,\Sim(1))$ and $(\Enc,\Sim(0))$. {This prevents from coming to agreement for the choice of the reduction.} Therefore, this method only works for proving the sparse lossiness of one-sided encodings.

{

\smallskip\noindent
\textit{From negligible $\cuteness$ to negligible privacy.} 
When $\lambda$, the $\cuteness$ quantity, is sufficiently small, we construct another encoding different from the one obtained by disguising lemma.
Let $\Pi \in \SPL[\lambda,\mu]$ for some $\lambda \ll 1$ and $R$ be its corresponding $\cute$ reduction. Moreover, take as advice $x_1^* \in \Pi_Y \cap \{0,1\}^n$ and $x_0^* \in \Pi_N \cap \{0,1\}^n$ (in the one-sided variant, only $x_1^*$ suffices). The randomized encoding for $\Pi$ is as follows: 
\begin{equation*}
\Enc(y) := R\left(y\right) \quad \text{and} \quad \forall b \in \{0,1\}: \Sim(b) := R\left(x_b^*\right) \, .
\end{equation*}
The correctness transfers immediately by definition.
For computing the privacy, we proceed as follows. Fix an input $y$ and let $b:=\chi_\Pi(y)$. We must bound $\|\Enc(y) - \Sim(b)\|$ from above. Let $X$ be a random variable that is $y$ with probability $1/2$ and $x^*_b$ otherwise. Note that $X$ has a uniform distribution over $2$ elements. Therefore, the $\cuteness$ guarantee that $I(X;R(X)) \leq \lambda$. Let $M$ be a random variable that is $\Enc(y)$ with probability $1/2$ and $\Sim(b)$ otherwise. Then, we have that
\begin{align*}
I(X;R(X)) &= \mathbb{E}_{X}[D_{\mathrm{KL}}(p_{R(X)|X}\|p_{R(X)})] = \frac{1}{2} D_{\mathrm{KL}}(p_{\Enc(y)} \| p_M) + \frac{1}{2} D_{\mathrm{KL}} (p_{\Sim(b)} \| p_M) \, .
\end{align*} 
The right-hand side can be bounded from below by $\|\Enc(y) - \Sim(b) \|^2/(2 \ln 2)$
according to Pinsker's inequality and the fact that $\|\Enc(y) - M\| = \|\Sim(b) - M\| = \|\Enc(y) - \Sim(b)\|/2$. Therefore, if $\lambda(n)=\negl(n)$, the privacy of the encoding would be negligible. In other words, we obtain
\begin{align}\label{eq:intro-negl}
\SPL[\negl(n),\mu(n)] &\subseteq \SRE[\negl(n),\mu(n)]/\poly(n) \, , \\
\OSPL[\negl(n),\mu(n)] &\subseteq \OSRE[\negl(n),\mu(n)]/\poly(n)\, . \nonumber
\end{align}

Finally, combining Equations~\eqref{eq:intro-log},~\eqref{eq:intro-re-to-sl}, and~\eqref{eq:intro-negl} yield Theorem~\ref{intro:sl-re-th}.

\medskip\noindent
\textbf{Proof of Theorem~\ref{intro:thm-niszk} (amplifying $\NISZK$).} 
The proof is a direct application of Key Lemma. Applebaum and Raykov implicitly (see \cite[Theorme~7 \& 8]{AR16}) show that for every integer $k \geq 1$, it holds that $\NISZK[c,s,\varepsilon]\subseteq \OSRE[k(c+\varepsilon),(2s)^k+k(c+\varepsilon)]$. By taking $k=1$ and imposing the condition $2s(n)+c(n)+\varepsilon(n) = 1-1/\poly(n)$ the former collapses to $\OSRE[1-1/\poly(n),1-1/\poly(n)]$ which itself is inside $\OSRE[\negl(n),\negl(n)]$ by Key Lemma. Finally, note that $\OSRE[\negl(n),\negl(n)]$ is equivalent to $\NISZK$ up to some polynomial-size advice~\cite[Therorem~1]{AR16}. This completes the proof.
}

\medskip\noindent
\textbf{Proof sketch of Theorem~\ref{intro:error-rem-thm} (from $\SZK$ to OWFs via error removal of $\OSRE$).}  
Let~$\Pi \in \OSRE[d(n),0]$ be a worst-case-hard problem with encoding~$\Enc$ and one-sided simulator~$\Sim$.\footnote{Note that since $\Pi \in \OSRE$, the simulator $\Sim$ works only on YES instances. We therefore omit $\chi_\Pi(x)$ from the input of $\Sim$, as it is always 1 in this case.} 
The circuits $\Enc$ and~$\Sim$ are randomized and take random coins as input. Without loss of generality, we can assume that they both use $\kappa$ bits of randomness. Note that $\kappa=\omega(\log n)$ since otherwise $\Pi \in \BPP$.
{
 Let $C_0(\cdot)$
and $C_1[y](\cdot)$ respectively denote $\Sim(\cdot)$ and $\Enc(y;\cdot)$ whose randomness is considered as the input. Let $\mathsf{F}:\{0,1\}^\kappa \times \{0,1\} \rightarrow \{0,1\}^*$ be the following function: 
$$ \mathsf{F}(b, r) := \begin{cases} C_0(r) & b=0 \\ C_1[x](r) & b=1 \end{cases} \, $$
where $x$ is a fixed YES instance of $\Pi$.
We show that $\mathsf{F}$ is a weak OWF, meaning that the probability of any polynomial-time adversary successfully inverting it on a random $\kappa$-bit input is at most $3/4$. This implies the existence of OWFs via Yao's hardness amplification~\cite{Yao82}.}
Assume that there exists an inverter $\mathcal{A}$ 
for $\mathsf{F}$ that succeeds with probability larger than $3/4$. We use $\mathcal{A}$ to decide any given instance $y$ of $\Pi$ as follows: Compute $C_1(r)$,
and $C_1[y](r)$ for a random value $r$, and then return $C_b(r)$ to $\mathcal{A}$. If $b = 0$,
then $\mathcal{A}$ receives an instance of the function $\mathsf{F}$ and can therefore invert
it. This does not help us with solving $\Pi$. Conversely, when $b = 1$, if $y$
is a YES instance, then by the privacy of the encoding, the statistical distance between $C_1[y]$ and $C_1[x]$ would be at most $2d$. Therefore, $\mathcal{A}$ would succeed to invert it with probability $> 3/4-2d$. On the other hand,
when $b = 1$, if $y$ is a NO instance, then $C_1[y]$ and $C_1[x]$ have distinct support. This is because the encoding has perfect correctness. Therefore, if $b = 1$
and $y$ is a NO instance, then there would be no preimage for the value that $\mathcal{A}$ aims to invert, which results
in $\mathcal{A}$ failing. Now, assume that $d \leq 1/4$. Then we have $3/4-2d \geq 1/4$ and, therefore, there is a constant gap between the success rate of $\mathcal{A}$ for YES and NO instances of $y$. By repeating this test for several times, we can decide $y$ by
observing the success rate of $\mathcal{A}$. Therefore, $\mathsf{F}$ is a weak OWF in this case.

So far, we have required that $d \leq 1/4$. Fortunately, this limitation can be lifted according to Key Lemma since $\OSRE[1-1/\poly(n),0] \subseteq \OSRE[1/4,0]$. Therefore the worst-case hardness of $\OSRE[1-1/\poly(n),0]$ implies OWFs.
On the other hand, note that~$\OSRE[\negl(n),\negl(n)] = \NISZK$~\cite{AR16} (up to non-uniform advice) and that $\NISZK$ is worst-case hard if and only if $\SZK$ is worst-case hard~\cite{GSV99}. Hence, the worst-case hardness of $\SZK$ implies the worst-case hardness of $\OSRE[\negl(n),\negl(n)]$ and if $\OSRE[\negl(n),\negl(n)] \subseteq \OSRE[1/2-1/\poly(n),0]$, then OWFs exist.

\medskip
\noindent\textbf{Proof sketch of Theorem~\ref{intro:io} (from weak \& imperfect iO to OWFs).}
To obtain the theorem, we analyze the {statistical behavior} of obfuscation. An algorithm is said to be $\alpha(n)$-statistical obfuscation (sO)
if it satisfies the following indistinguishability property: the obfuscations of any two functionally-equivalent circuits on $n$-bit inputs are at most $\alpha(n)$-statistically far. We show that if sO with distance $\alpha(n) = 1-1/\poly(n)$ and error $\varepsilon(n) = 1/2 - 1/\poly(n)$ exists, then {$\NP \subseteq \coNPpoly$}, i.e., the Polynomial Hierarchy collapses to its third level.\footnote{Recall that an obfuscation $\mathsf{O}$ has error $\varepsilon(n)$ if for every circuit over $n$-bit inputs, it holds that $\Pr[\exists x \in \{0,1\}^n: \mathsf{O}(C)(x) \neq C(x)] \leq \varepsilon(n)$.}

The key idea is to leverage sO
to build a randomized encoding for the unique satisfiability ($\uSAT$) problem. This problem asks to decide whether a $3\textsc{CNF}$ formula over~$n$ variables has a satisfiable assignment under the promise that it has at most one of such. Works of~\cite{KMN+14,IL25} observe that sO gives a worst-case to average-case reduction for $\uSAT$. We follow the same blueprint. {More precisely, if a circuit}~$C$ is a YES instance of~$\uSAT$, {i.e., it evaluates to 1 on exactly one input,} then a random shift of~$C$, namely the circuit~$C_z(x) = C(x \oplus z)$, where $z \rnd \{0,1\}^n$, has a truth-table whose distribution is identical to that of a random point function
. Therefore,~$\text{sO}(C_z)$ is statistically close to~$\text{sO}(\mathds{1}_{x=x^*})$, where $x^* \rnd \{0,1\}^n$. 
{This gives us a one-sided randomized encoding for $\uSAT$, with {encoding $\Enc(C;r,z) = \text{sO}(C_z;r)$ and simulator $\Sim(1;r,z)=\text{sO}(\mathds{1}_{x=z};r)$} (which works for YES instances).} 
As a result~$\UP \subseteq \OSRE[\alpha,\varepsilon]$, where~$\alpha$ is the statistical indistinguishability bound of sO and~$\varepsilon$ is the error.\footnote{This part must be done carefully and requires a change of variable from~$n$, the size of the input, to~$|C|$ the instance size since the class~$\OSRE$ is parameterized by the instance size. Also note that~$\uSAT$ is complete for~$\UP$.} By using Key Lemma, we can amplify the privacy and the correctness which leads to the following:
$$\UP \subseteq \OSRE[\alpha,\varepsilon] \subseteq \OSRE\left[ \negl(n),\negl(n)\right]/\poly = \NISZK/\poly \, .$$

 Hence, $\UP$ collapses to $ \SZK/\poly$. {Finally, via the Valiant-Vazirani polynomial-time reduction of $\SAT$ to $\uSAT$, this sequence implies that $\SAT \in \SZK \slash \poly$, meaning that the Polynomial Hierarchy collapses to its third level.}
 
 In fact, a two-sided randomized encoding for $\uSAT$ can be obtained by the exact same idea of leveraging the obfuscation. This is because for a non-satisfiable circuit $C$, we have that $\text{sO}(C;r) \approx \text{sO}(Z;r)$ where $Z$ is the zero function. However, the rest of the argument does not hold due to lack of a privacy amplification and error reduction for two-sided encodings.

We now draw the final conclusion of this part. The preceding analysis implies that, unless the PH collapses, any ($\alpha$,$\varepsilon$)-sO requires either $\alpha$ to be negligibly close to~$1$, or~$\varepsilon$ to be negligibly close to~$1/2$. Now consider an iO with computational distance $\tau(n)=1-1/\poly(n)$ and error $\varepsilon(n)=1/2-1/\poly(n)$. It follows that for any such iO, there exist two families of circuits $\{C_n\}_n$ and $\{C'_n\}_n$ such that their obfuscated variants are statistically far, as the statistical distance has to be $\alpha(n)=1-\negl(n)$ (or else, the PH collapses), and computationally close, as we assumed the computational distance to be $\tau(n)=1-1/\poly(n)$. This statistically-far-but-computationally-close pair of distributions is an \emph{EFI pair}, which is known to imply OWFs~\cite{Gol90,NR06,BDRV19}. Putting everything together, assuming the existence of iO with computational distance $\tau(n)=1-1/\poly(n)$ and error $\varepsilon(n)=1/2-1/\poly(n)$ implies the existence of OWFs unless the PH collapses to its third level.

\subsection{Other Related Works}
Hirahara and Nanashima~\cite{HN24} study the existence of OWFs by leveraging Kolmogrov complexity tools. They show that if~$\NP$ is polynomial-time worst-case hard and moreover~$\NP \subseteq \CZK$ then OWFs exist. The complexity class~$\CZK$ is the computational variant of~$\SZK$ where the zero-knowledge simulator offers only computational indistinguishability guarantees. In fact the condition is equivalent to the polynomial-time worst-case hardness of~$\CZK$ and the membership of the decision variant of the minimum circuit size problem to this class. \footnote{The search variant of this problem roughly asks to find the minimum-size circuit that outputs the truth-table of a given function.}

Chakraborty, Hulett, and Khurana~\cite{CHK25} show how to build a OWF from the worst-case hardness of~$\NP$ if $\NP$ is contained in~$\NICZK$ (computational zero-knowledge) with~$\varepsilon_{zk}+\sqrt{\varepsilon_{s}} < 1$, where~$\varepsilon_{zk}$ and~$\varepsilon_{s}$ are respectively zero-knowledge and soundness errors. Their result is based on the work of Hirahara and Nanashima~\cite{HN24}.

IV-OWSGs seem to be weaker than OWFs: there are (quantum or classical) oracles relative to which IV-OWSGs exist while OWFs do not (see \cite{Kre21,KQST23}).
Khurana and Tomer~\cite{KT25} show the existence of IV-OWSGs with mixed outcome based on the worst-case hardness of~$\mathsf{P}^{\# \mathsf{P}}$ and an extra ad-hoc assumption. More generally, it is known that IV-OWSGs with mixed states are equivalent to EFIs and therefore the existence of the latter implies the existence of the former~\cite{KT24}. Although pure IV-OWSGs imply EFIs~\cite{KT24}, the other direction is not known yet.  Bostanci et al.~\cite{Bos24} construct IV-OWSG with pure outcome via the worst-case hardness of the so-called Hamiltonian Phase State Problem.

\subsection{Open Questions}
The most important question that remains open is an error removal for~$\OSRE$. As we showed, such an amplification would immediately implies OWFs from the worst-case hardness of~$\SZK$. Even a weaker amplification such as the following one yields the same result: if we can efficiently (up to some advice) recognize the instances or/and the random coins on which the encoding errs. Also working with the sparsely lossy characterization of~$\OSRE$ might be helpful. We leave this question to the future work.

\section{Preliminaries}\label{sec:preliminaries}
In this work, we always consider non-uniform algorithms. All classical algorithms are quantum algorithms, therefore, we mostly use the quantum formalism for generalization and simplification. When the distinction is necessary, we explicitly mention it in the beginning of a section or inside a statement.

\medskip\noindent\textbf{Notation.} We let~$n$ denote the security parameter, and all variables are implicitly parameterized by~$n$. We let~$\mathsf{MS}_n$ denote the set of all mixed states over~$n$ qubits and we define~$\mathsf{MS}_* := \cup_{n=1}^{\infty}\mathsf{MS}_n$. For a positive integer~$n$, we let~$[n]$ denote~$\{1,2,\cdots,n\}$. The set of all permutations over~$[n]$ is~$\mathfrak{S}_n$. We abuse the notation and use the same symbol to refer to the uniform distribution over all permutations of~$[n]$. The set of natural numbers $\{1,2,3,\cdots\}$ is denoted by $\mathbb{N}$. We denote by $\mathbb{R}^+$ the set of positive real numbers.\\  
A function~$f(n)$ is~$\poly_\ell(n)$ if~$f(n) = O(n^\ell)$ (we drop~$\ell$ when the degree is not specified), is~$\negl(n)$ if~$f(n) = 1/n^{\omega(1)}$, and~$\subexp(n)$ if~$f(n) = 2^{O(n^c)}$ for some constant~$c < 1$.

\medskip\noindent\textbf{Uniform and $s$-Uniform Distributions.} For any set~$S$, we let~$\unif_{S}$ denote the uniform distribution over~$S$. A distribution is called~$s$-uniform if it is uniform over a multiset of~$s$ elements.

\medskip\noindent\textbf{Promise Problems.} A Promise Problem $\Pi$ consists of two disjoint sets $\Pi_Y$ and~$\Pi_N$ in~$\{ 0,1 \}^*$, respectively referred to as the set of YES and NO instances. Problem $\Pi$ asks to decide whether a given instance, which is promised to lie in $\Pi_Y \cup \Pi_N$, belongs $\Pi_Y$ or $\Pi_N$.
\begin{definition}[Characteristic Function of a Promise Problem]
For a promise problem $\Pi$, its characteristic function is the map $\chi_\Pi(x) : \{0,1\}^* \rightarrow \{0,1,\star\}$ given by 
\begin{equation*}
	{\chi}_\Pi(x) = 
	\begin{cases}
		1 & \text{if } x \in \Pi_Y\\ 
		0 & \text{if } x \in \Pi_N\\
		\star & \text{otherwise}
	\end{cases}~.
\end{equation*}
\end{definition}
In this work, we consider worst-case hardness as introduced below.

\begin{definition} A promise problem~$\Pi$ is worst-case hard, if for any (randomized or quantum depending on the context) polynomial-time algorithm~$\mathcal{A}$ and any sufficiently large~$n \in \mathbb{N}$ , there exists an input~$x \in (\Pi_Y \cup \Pi_N ) \cap \{0,1\}^n$ such that~$\Pr[\mathcal{A}(x) = \chi_\Pi(x)] < 2/3$. 
\end{definition}

\medskip\noindent\textbf{Classical information.} Given two probability distributions~$X$ and~$Y$ over~$\Sigma$, their statistical distance, also called total variation distance, is defined as
$$\Delta(X,Y):=\frac{1}{2}\sum_{x\in\Sigma}|\Pr(X=x)-\Pr(Y=x)|~.$$
The Kullback–Leibler divergence or classical relative entropy of~$X$ with respect to~$Y$ is defined as
$$D_{\textrm{KL}}(X\|Y):=\sum_{x\in\Sigma} \Pr(X=x)\log(\frac{\Pr(X=x)}{\Pr(Y=x)})~.$$

\begin{lemma}[Chernoff-Hoeffding Bound]\label{lemma:hoeff}
Let~$X_1,X_2,\cdots,X_k$ be mutually independent random variables in~$[0,1]$ and~$\mu:= \mathbb{E}[\sum_{i=1}^{k} X_i ]$. For every~$t>0$, it holds that $\mathrm{Pr}[|\sum_{i=1}^{k} X_i - \mu| > t] \leq 2 \exp(-2t^2/k)$.
\end{lemma}

\medskip\noindent\textbf{Quantum information.}
For any two mixed states~$\rho$ and~$\sigma$, we let $D_{\textrm{tr}} := \Tr(|\rho-\sigma|)/2$ denote their trace distance. 
Let~$\map:\{0,1\}^{n} \rightarrow \mathsf{MS}_{m}$ be any quantum mapping and~$X$ a random variable supported over~$\{0,1\}^{n}$. We let
\begin{align}\label{eq:c-q_map}\rho_{X,\map(X)} \eqdef \sum\limits_{x \in \{0,1\}^{n}} \Pr_X(x) \op{x}{x} \otimes \map(x) \ .\end{align}

For a mixed state~$\rho$, we let~$S(\rho)\eqdef \mathrm{Tr}(\rho \log_2 \rho)$ denote the Von Neumann entropy of~$\rho$.
The quantum mutual information of two subsystems~$A$ and~$B$ is defined as follows. Let~$\rho_{AB}$ be their joint state, then
$$I_q(A;B)_\rho \eqdef S(\rho_A) + S(\rho_B) - S(\rho_{AB}) \ ,$$
where~$\rho_A = \tr_{B}(\rho_AB)$ and~$\rho_B = \tr_{A}(\rho_AB)$. For the sake of simplicity, we sometimes drop the subscripts~$q$ and~$\rho$ in~$I_q$. When working with quantum systems~$A,B$, the notation~$I(A;B)$ implicitly refers to~$I_q(A;B)$.

The following lemma states that if the outcome of a measurement is close to deterministic, then it must not alter much the state.
\begin{lemma}[Gentle Measurement Lemma~\cite{winter99}]\label{lemma:gml}
	Let~$\rho$ be a mixed state and~$\{\Lambda,I-\Lambda\}$ a two-outcome POVM with $\tr(\Lambda\rho)\geq1-\eps$, then~$D_{\textrm{tr}}(\rho,\rho')\leq\sqrt{\eps}/2$, where $\rho'=\frac{\sqrt{\Lambda}\rho\sqrt{\Lambda}}{\tr(\Lambda\rho)}$.
\end{lemma}

For two quantum states \(\sigma,\rho\) stored in two different registers \(A,B\), the swap test is executed on the registers \(A,B\) and a control register \(C\) initialized to \(\ketbra{1}\).
It applies Hadamard on \(C\), swaps \(A\) and \(B\) conditioned on \(C\), and measures \(B\) on the Hadamard basis.
\begin{lemma}[SWAP Test~\cite{fingerprinting}]\label{lemma:swap}
	The SWAP test on input \((\sigma,\rho)\) outputs 1 with probability~$(1+\tr(\rho\sigma))/2$, in which case we say that it passes the test. For pure states \(\ket\sigma,\ket\rho\), it equals to \((1+|\braket{\rho}{\sigma}|^2)/2\).
\end{lemma}
Given that the trace distance of two pure states \(\ket\sigma,\ket\rho\) can be expressed in terms of the inner product uniquely as $\sqrt{1-|\braket{\rho}{\sigma}|^2}$, the SWAP test can also be used to calculate their trace distance.

\medskip\noindent\textbf{Distance notation.}
To avoid repeated proofs, we use the following definition to unify the distance analysis in both classical and quantum settings.
\begin{definition}\label{def:ell1}
We use the notation~$\|X-Y\|$ to refer to the statistical distance~$\Delta(X,Y)$ when variables~$X,Y$ are classical distributions, or the trace distance when they are quantum states.
\end{definition}

We will use the following lemmas.

\begin{lemma}[{Direct Product Lemma \cite{wat02,SV03}}]\label{lemma:direct-product}
Let~$X,Y$ be two quantum states or classical distributions such that~$\|X-Y\| = \varepsilon$. For any~$k$, it holds that 
$$  1-2\exp\left(\frac{-k\varepsilon^2}{2} \right)  \leq \left\|X^{\otimes k} - Y^{\otimes k}\right\| \leq k \varepsilon \, .$$
\end{lemma}

\begin{lemma}[{XOR Lemma\cite{wat02,SV03}}]\label{lemma:xor}
Let $\mathsf{XOR}$ be a function that maps the tuple $(X_0,X_1,1^k)$ to~$(Z_0,Z_1)$ as follows:
\begin{itemize}
\item[] $Z_0:$ Uniformly sample~$(b_1,\cdots,b_k) \in \{0,1\}^k$ such that~$b_1 \oplus \cdots \oplus b_k = 0$ and output a sample from $X_{b_1} \otimes \cdots \otimes X_{b_k}$.
\item[] $Z_1:$ Uniformly sample~$(b_1,\cdots,b_k) \in \{0,1\}^k$ such that~$b_1 \oplus \cdots \oplus b_k = 1$ and output a sample from $X_{b_1} \otimes \cdots \otimes X_{b_k}$.
\end{itemize}
Then $\|Z_0-Z_1\| = \|X_0-X_1\|^k$.
\end{lemma}

\begin{lemma}\label{lemma:xor-trick}
Let~$(X_0,X_1)$ and $(Y_0,Y_1)$ be such that $\|X_0 - Y_0\|\geq \varepsilon$ and $\|X_1 - Y_1\| \geq \delta$. Let $(P, \cdot):=\mathsf{XOR}(X_0,X_1,1^k)$ and $(Q, \cdot):=\mathsf{XOR}(Y_0,Y_1,1^k)$ for some $k \geq 2$. Then~$\|P-Q\| \geq (\varepsilon+\delta)/2$.
\end{lemma}
\begin{proof}
Consider the first component of $P$. The marginal distribution of~$b_1$ (see Lemma~\ref{lemma:xor}) is uniform. Thus with half probability $b_1=0$ and the first component is selected from $X_0$. Same holds for~$Q$: when~$b_1=0$, its first component is~$Y_0$. Note that the distance is bounded by~$\varepsilon$ in this case. Same holds for~$b_1=1$ with the lower bound~$\delta$. This completes the proof. 
\end{proof}

\medskip\noindent\textbf{Permutation-invariant function.} 
A function is permutation-invariant if applying any permutation to its inputs does not alter its output value. This includes a wide range of functions, such as~$\AND,\OR$, $\MAJ$, and~$\PARITY$.
We use the following technical lemma about non-constant permutation-invariant functions. 
\begin{lemma}\label{lemma:minimum-p}
Let~$f:\{0,1\}^m \rightarrow \{0,1\}$ be a non-constant permutation-invariant function. Then there exists an integer~$1 \leq p \leq m$ such that
$$f(\ \underbrace{1,1,\ldots,1}_{p-1},0,0,\ldots,0) = 0\, , \quad \text{and} \quad f(\ \underbrace{1,1,\ldots,1}_{p},0,0,\ldots,0) = 1 \ .$$ 
We let~$p(f)$ denote the minimum choice of such an integer.
\end{lemma}
\begin{proof}
The set~$\{0,1\}^m$ can be partitioned into~$m+1$ equivalence classes where each class consists of strings with the same number of~$1$'s. We note that the result of a permutation on an input falls in the same equivalence class. Therefore, since the function is permutation-invariant, then the evaluation of~$f$ over each input is determined by its class. Because the function is non-constant, there must exist two consecutive classes (the classes can be ordered by the number of~$1$'s that they represent) with different evaluation under~$f$. This completes the proof. 
\end{proof}

\medskip\noindent\textbf{Complexity class~$\SZK$.} 
There are various equivalent definitions of the complexity class~$\SZK$. We start with the standard one.
\begin{definition}[Class $\SZK$]
For any~$c,s,\varepsilon:\mathbb{N} \rightarrow [0,1]$
we let~$\SZK[c(n),s(n),\varepsilon(n)]$ be the set of all promise problems~$\Pi$ that admit a public coin interactive proof system~$(\mathcal{P},\mathcal{V})$ such that:
\begin{itemize}
\item \textbf{Completeness}: For any~$x\in \Pi_Y \cap \{0,1\}^n$, it holds that $\Pr[\langle \mathcal{P}, \mathcal{V} \rangle(x)=1] \geq 1-c(n)$. 
\item \textbf{Soundness}: For any~$x\in \Pi_N \cap \{0,1\}^n$ and every prover $\mathcal{P}^*$, it holds that~ \linebreak$\Pr[\langle \mathcal{P}^*, \mathcal{V} \rangle(x)=1] \leq s(n)$. 
\item \textbf{Zero-Knowledge}: There exists a polynomial-time algorithm~$\Sim$ such that for every~$x\in \Pi_Y \cap \{0,1\}^n$, it holds that
$$\|\Sim(x) - \view_{\mathcal{P},\mathcal{V}}(x)\| \leq \varepsilon(n)\, .$$
\end{itemize}
We let~$\SZK:= \SZK[\negl(n),\negl(n),\negl(n)]$.
\end{definition}

For certain parameters, the following problem is complete for $\SZK$.
\begin{definition}[Statistical Difference]\label{def:sd}
	Let~$\alpha ,\beta \in [0,1]$ such that~$\alpha < \beta$. Given two circuits~$\mathcal{C}_0$ and~$\mathcal{C}_1$, with the promise that either~$\|C_0-C_1\| \geq \beta$ (corresponds to NO instances) or~$\|C_0-C_1\| \leq \alpha$ (corresponds to YES instances). The~$\SD_{\alpha,\beta}$ problem is to decide which one is the case. When it is clear from the context, we abuse the notation and let~$C$ denote its outcome distribution too. 
\end{definition}

Sahai and Vadhan~\cite{SV03} show that~$\SD_{1/3,2/3}$ is complete for~$\SZK$. 
They also show that~$\SZK$ has a the following closure property.
\begin{lemma}[{From~\cite[Th.~4.9]{SV03}}]\label{lemma:f(SZK)}
Let~$f:\{0,1\}^m \rightarrow \{0,1\}$ be a Boolean function. For a problem~$\Pi$ define $f(\Pi)$ as follows:
$$f(\Pi)_Y := \{(x_1,\cdots,x_m) \ | \ f(\chi_\Pi(x_1),\cdots, \chi_\Pi(x_m)) = 1\}\, ,$$
and
$$f(\Pi)_N := \{(x_1,\cdots,x_m) \ | \ f(\chi_\Pi(x_1),\cdots, \chi_\Pi(x_m)) = 0\}\, .$$
For every~$\Pi \in \SZK$, it holds that~$f(\Pi) \in \SZK$.
\end{lemma}

We also have the following lemma about~$\SZK$ that is derived from~\cite[Thm.4.11]{Dru15} and~\cite[Thm.~4]{BBF16}.
\begin{lemma}\label{prelim:lemma-szk-derand}
It holds that~$\SZK/\poly \subseteq \NP/\poly \cap \mathsf{co}\NP/\poly$.
\end{lemma}

\medskip\noindent\textbf{Complexity class $\NISZK$.} It is the class of problems that have non-interactive zero-knowledge proofs.
\begin{definition}[Class~$\NISZK$]
It is defined similarly to~$\SZK[c(n),s(n),\varepsilon(n)]$ with the following differences: (1) Only one single message from the prover~$\mathcal{P} $ to the verifier~$\mathcal{V}$ is communicated. (2) There exists a polynomial-time algorithm~$\mathsf{Deal}(1^n)$ that is run before the communication to obtain a uniformly random string~$\sigma$. The parameter~$\sigma$ is given to both~$\mathcal{P}$ and~$\mathcal{V}$ and is considered as part of the view of the protocol. We also let $\NISZK:=\NISZK[\negl(n),\negl(n),\negl(n)]$.
\end{definition}

The worst-case hardness of~$\NISZK$ and~$\SZK$ are equivalent. 
\begin{lemma}[{Adapted from \cite[Th.~1.5]{GSV99}}]\label{lemma:gsv99}
It holds that~$\NISZK \subseteq \mathsf{P}/\poly$ if and only if~$\SZK \subseteq \mathsf{P}/\poly$. 
\end{lemma}

The randomized encodings are inside~$\NISZK$ and~$\SZK$. 
\begin{lemma}[From \cite{AR16}]\label{lemma:re-is-szk}
It hold that~$\SRE \subseteq \SZK$,~$\OSRE \subseteq \NISZK$, and $\NISZK \subseteq \OSRE/\poly$. 
\end{lemma}

\medskip\noindent\textbf{Polynomial Hierarchy.}
We let~$\Sigma_0^p := \Pi_0^p := \mathsf{P}$ be the class of problems solvable in polynomial time. The levels of the Polynomial Hierarchy $\PH$ are defined as follows: for every~$i\in \mathbb{N}$, we let
$$\Sigma_{i+1}^p := \NP^{\Sigma_i^p} \text{  and  } \Pi_{i+1}^p := \mathsf{co}\NP^{\Sigma_i^p}\, .$$ 
It is a common belief that $\PH$ does not collapse to any of its levels, namely, it is unlikely to have~$\PH=\Sigma_i^p = \Pi_i^p$ for any~$i \geq 1$. 
The following result will be used in our theorems.
\begin{lemma}[From {\cite{Yap83}}]\label{lemma:PH}
If~$\NP \subseteq \mathsf{co}\NP/\poly$, then~$\PH=\Sigma_3^p = \Pi_3^p$.
\end{lemma}

\medskip\noindent\textbf{Cryptographic primitives.}
One-way functions are defined as follows:
\begin{definition}[One-Way Functions]
A family of non-uniform PPT algorithms~$\mathsf{F}:=\{\mathsf{F}_n\}_{n \in \mathbb{N}}$ is said to be one-way (OWF) if for all sufficiently large~$n$ and any polynomial-time algorithm~$\mathcal{A}$, it holds that
$$\Pr_{x \sim \mathcal{U}_{\{0,1\}^n}}[\mathsf{F}( \mathcal{A}( \mathsf{F}(x) ) ) = \mathsf{F}(x)] \leq \negl(n) \, .$$
Furthermore, we say that~$\mathsf{F}$ is weak-one-way if the above probability is bounded by~$1-1/\poly(n)$.
\end{definition}

It is shown by~\cite{Yao82} that weak one-way functions imply one-way functions.

In this work, we consider the inefficient-verifier one-way state generators with pure output.

\begin{definition}[One-Way State Generators]\label{def:owsg}
	A one-way state generator (OWSG) is a tuple of algorithms $\mathsf{G}:=(\kgen,\sgen,\ver)$ with the following specification:
	\begin{itemize}
		\item $\kgen_h(1^n)\rightarrow k$: is a QPT algorithm that given the security parameter~$1^n$ and a~$\poly(n)$-size classical advice~$h$, outputs a classical string $k\in\{0,1\}^n$;
		\item $\sgen(k)\rightarrow \ket{\psi_k}$: is a QPT algorithm that given a classical string~$k$, outputs an $m$-qubit quantum state;
		\item $\ver(k,\ket{\psi}) \in \{0,1\}$: is a (possibly unbounded) algorithm that given a classical string $k$ and a quantum state $\ket{\psi}$ outputs either $0$ or $1$.
	\end{itemize}
	Further, they satisfy the following properties:
	\begin{enumerate}
		\item \textbf{Correctness.} Outputs of the samplers $(\kgen,\sgen)$ pass the verification with overwhelming probability, i.e.,
		\begin{equation*}
			\pr[\substack{k\leftarrow\kgen_h \\ \ket{\psi_k}\leftarrow\sgen(k)}]{\ver(k,\ket{\psi_k})=1}\geq 1-\negl(n)\, .
		\end{equation*}
		\item \textbf{Security.} For every non-uniform~$T(n)$-time adversary $\adv$, and any polynomial $t(n)$
		\begin{equation*}
			\pr[\substack{k\leftarrow\kgen_h \\ \ket{\psi_k}\leftarrow\sgen(k) \\ k'\leftarrow\adv(\ket{\psi_k}^{\otimes t};h)}]{\ver(k',\ket{\psi_k})=1}\leq \negl(n) \, .
		\end{equation*}
	\end{enumerate}
	
	Furthermore, we say that~$\mathsf{G}$ is a weak-OWSG if the inequality concerning security (Property 2) is bounded by~$1-1/\poly(n)$.

\end{definition}
It is shown in~\cite{MY24} that weak OWSGs imply OWSGs.

\section{Complexity Classes}\label{sec:spl}

In this section, we define useful complexity classes. We begin by the class of problems that have ``lossy'' reductions. In our analysis, we consider both randomized functions and quantum mappings. All the statements hold with respect to both cases. 
We explicitly highlight the distinction when the analysis requires to distinguish between the two cases.

Classically, a randomized function~$R:\{0,1\}^{*} \rightarrow \{0,1\}^*$ is said to be~$\lambda$-lossy for a distribution~$X$ if~$I(X;R(X)) \leq \lambda$. 
Below, we also consider general mappings with classical input and quantum output.

\begin{definition}[Lossy Mapping on a Fixed Distribution]\label{def:lossy-mapp2}
	Let~$\lambda,m\geq 0$. Let~$\map:\{0,1\}^{*} \rightarrow S$ be a mapping, where $S = \{0,1\}^*$ (classical mapping) or $S = \mathsf{MS}_{*}$ (quantum mapping). We say that $\map$ is~$\lambda$-lossy for an~$m$-tuple distribution~$X=(X_1,X_2,\ldots,X_m)$ over $\{0,1\}^*$, if it holds that
	$$I(X;\map(X)) \leq m\lambda \ .$$
	For the sake of simplicity, we say that~$\map$ is $\lambda$-lossy, if it is~$\lambda$-lossy for all $m$-tuple distributions.
	
\end{definition}

As we will see later, it is useful to limit the input distributions to have a particular ``splitting'' structure.
\begin{definition}[Splitting Distribution]\label{def:split-lossy}
	Let~$S_0,S_1\subseteq\{0,1\}^*$ be two disjoint sets.
	We say that a distribution $X=(X_1,\ldots,X_m)$ splits over the pair~$(S_0,S_1)$ if for each $i\in[m]$, either~$\supp(X_i)\subseteq S_0$ or~$\supp(X_i)\subseteq S_1$. 
\end{definition}

An example of such distributions is~$(X,Y)$ where $X$ and $Y$ have disjoint supports.

In this work, we consider distributions that are not only sparse but also are uniform. 
\begin{definition}[Pairwise Independent $B$-Uniform Distribution]\label{def:B-uniform}
We say that a distribution $X=(X_1,\ldots,X_m)$ is pairwise-independent~$B$-uniform if $X_i$'s are independent and each of them follows a (not necessarily similar) uniform distribution over a multiset of size at most~$B$.
\end{definition}

For example, distribution~$(\mathcal{U}_n,\cdots,\mathcal{U}_n)$ is pairwise independent $2^n$-uniform.

For the sake of generality, we consider the class of problems that admit lossy $f$-reductions. Recall~$f$ reductions below.
\begin{definition}[$f$-Reduction]\label{def:red}
	Let~$n,m$ be positive integers, and~$\mu:\mathbb{N} \rightarrow [0,1/2]$ be a function of~$n$. Let~$f:\{0,1\}^m \rightarrow \{0,1\}$, and~$\Pi$ be a promise problem. An $f$-reduction for~$\Pi$ with error~$\mu$ 
	is a mapping~$\red:\{0,1\}^{mn} \rightarrow S$, where $S = \{0,1\}^*$ (classical) or $S = \mathsf{MS}_{*}$ (quantum), for which there exists an 
	unbounded algorithm ~$\mathcal{M}$,
	such that for all~$(x_1,\ldots,x_{m}) \in ((\Pi_Y \cup \Pi_N) \cap \{0,1\}^n )^{m}$, 
	\begin{align*}
		\Pr\left[ \mathcal{M}(\red(x_1,\ldots,x_{m})) = f\left(\chi_\Pi(x_1),\ldots,\chi_\Pi(x_{m}) \right)\right] \geq 1-\mu(n) \, ,
	\end{align*}
	where the probability is taken over the randomness of~$\red$ and~$\mathcal{M}$. \footnote{When considering quantum mappings, $\mathcal{M}$ can be a binary quantum measurement.}
\end{definition}

Note that reductions are classically defined from a problem~$\Pi$ to another problem~$\Sigma$ such that they map YES (and NO) instances of~$\Pi$ to YES (and NO) instances of~$\Sigma$. Any (possibly unbounded) solver for~$\Sigma$ can play the role of~$\mathcal{M}$ in our definition.

We introduce the class of \emph{\cute~problems} below.

\begin{definition}[Class $(\mathsf{Q})\SPL$]\label{def:cute}
Let~$n$ be the instance size. Let~$m \in \mathbb{N}$,~$\lambda \in \mathbb{R}^{+}$,and~$\mu \in [0,1/2]$ be all functions of~$n$ and~$f:\{0,1\}^m \rightarrow \{0,1\}$ be a Boolean function. 
	The class~$\SPL_f[\lambda,\mu]$ consists of all promise problems~$\Pi$ that has a polynomial-time $f$-reduction with error~$\mu$ with the following property:
\begin{itemize}
\item For every polynomial~$B(n)$ and for all sufficiently large~$n$, reduction $\red$ is~$\lambda(n)$-lossy (per Definition~\ref{def:lossy-mapp2}) for all distributions~$(X_1,\cdots,X_m)$ that split over~$(\Pi_N\cap \{0,1\}^n,\Pi_Y\cap \{0,1\}^n)$ (per Definition~\ref{def:split-lossy}) and are pairwise independent $B$-uniform (per Definition~\ref{def:B-uniform}).
\end{itemize}	
The class~$\QSPL_f$ is defined similarly by replacing classical reductions with quantum reductions. When~$f$ is the identity function, we remove the dependency on~$f$ from the notation.
\end{definition}

Roughly speaking, this class contains those problems with reductions that lose information on sparse uniform distributions that split over YES and NO instances. 

\begin{remark}[One-sided lossiness]
We also consider a specific type of $\SPL$ classes where the lossiness is only guaranteed for YES instances. We let~$\OSPL[\lambda(n)]$ be defined as above for~$f$ being the identity function where the reduction is~$\lambda(n)$-lossy for all distributions~$X$ over~$\Pi_Y\cap \{0,1\}^n$ and are~$B$-uniform. It trivially includes~$\SPL[\lambda(n)]$. The class~$\OQSPL$ is defined similarly.
\end{remark}

We also recall the class of randomized encodings.
\begin{definition}[Class $(\mathsf{Q})\SRE$ (from~\cite{AIKP15,AR16})]\label{def:class-re}
Let~$n$ be the instance size. Let~$d\in [0,1]$,and~$\mu \in [0,1/2]$ be both functions of~$n$.
	 The class~$\SRE[d,\mu]$ consists of all promise problems~$\Pi$ for which there exists a polynomial-time algorithm $\Enc$ that has the following properties: 
	\begin{itemize}
		
		\item (\textbf{$\mu$-Correctness}) There exists an algorithm $\Dec$ such that for all 
$x \in \{0,1\}^n$: $$\Pr\left[\Dec(\Enc(x)) \neq \chi_\Pi(x)\right] \leq \mu(n) \, .$$

		\item (\textbf{$d$-Privacy}) There exists a polynomial-time algorithm $\Sim$ such that for all 
$x  \in \{0,1\}^n$: $$\|\Sim(1^n,\chi_\Pi(x))-\Enc(x)\| \leq d(n)~,$$
		as distributions. For simplicity, we drop the argument $1^n$, when input size is clear from the context.
	\end{itemize}
The class~$\QRE[d,\mu]$ is defined similarly by replacing classical algorithm~$\Enc$ with quantum algorithms. We let~$\SRE$ and~$\QRE$ denote the variants where both~$d$ and~$\mu$ are negligible.
\end{definition}

The following inclusions are trivial. 
\begin{proposition}
It hold that~$\BPP \subseteq \SRE[0,\negl(n)]$ and $\BQP \subseteq \QRE[0,\negl(n)]$.
\end{proposition}

\begin{remark}[One-sided encodings]
The one-sided variant of randomized encodings are defined similarly except that the privacy is only guaranteed for YES instances and not all instances. More precisely, there exists a polynomial-time algorithm~$\Sim$ that is statistically close to~$\Enc(x)$ for every YES instance~$x$. We let~$\OSRE[d,\mu]$ and~$\OQRE[d,\mu]$ denote these classes in the classical and quantum settings, respectively. The one-sided variant is possibly a larger class. AppleBaum and Raykov~\cite{AR16} show that $\OSRE[\negl(n),\negl(n)]$ is equal to the class of problems that have non-interactive statistical zero-knowledge proofs ($\NISZK$).
\end{remark}

\begin{remark}[Pure-outcome encodings]
In the above definition of (1)$\QRE$, there is no restriction on the type of the output of $\Enc$; it can be a mixed or a pure state. We also define the more restrictive class (1)$\QRE_{\mathsf{pure}}$ to be the class of those quantum encodings whose $\Enc(\cdot;r)$ function returns a pure state for any fixed classical randomness $r$ (note that the expectation $\Enc(\cdot;r)$ over the choice of $r$ is not pure).
\end{remark}

\section{Sparsely Lossy Reductions Are Randomized Encodings}

In this section, we show that~$\SPL_f$ is included in~$\SRE$ up to some classical advice for a variety of Boolean functions~$f$ as below:
\begin{definition}[Permutation-Invariant Function]\label{def:prelim-p}
A Boolean function~$f:\{0,1\}^m \rightarrow \{0,1\}$ is said to be permutation-invariant if for every~$\pi \in \perm_m$, it holds that~$f(\pi(b_1,b_2,\ldots,b_m)) = f(b_1,b_2,\ldots,b_m)$.
\end{definition} 

The functions~$\AND,\OR$, and~$\MAJ$ that were considered in~\cite{FS08,Dru15,BBDD+20} are all permutation-invariant. Moreover, the functions~$\PARITY$ and~$\MOD_k$ are of this type as well as~$\Th_k$.

Main theorem of this section is the following:
\begin{theorem}\label{th:spl-in-re}
For any permutation-invariant function~$f$ and any polynomial $q(n)$, it hold that $\SPL_f[\lambda,\mu] \subseteq \SRE[\delta+1/q,\mu]/\poly$ and~$\QSPL_f[\lambda,\mu] \subseteq \QRE[\delta+1/q,\mu]/\poly$ where $\delta \eqdef \min \{ \sqrt{(\lambda \ln 2)/2}, 1- 2^{-\lambda-2}\}$.
The advice size itself is a polynomial in~$q$.
\end{theorem}

The technical core for this theorem is an extended variant of \emph{disguising lemma}.
The disguising lemma first appeared in~\cite{Dru15} and was improved by~\cite{BBDD+20}. Those variants rely on lossiness with respect to all possible input distributions. Instead, we show that the lemma remains almost intact for maps that are lossy only over splitting uniform distributions with sparse support.
This is obtained by a more refined analysis but yet very similar to those of~\cite{Dru15,BBDD+20}. The proof is deferred to Appendix~\ref{app:disguising-proof}.

\begin{lemma}[Extended Disguising Lemma]\label{lemma:perm-dis-dis}
Let~$n,m,m_0,m_1$ be positive integers such that~$m = m_0+m_1 + 1$, and~$\red:\{0,1\}^{*} \rightarrow \mathsf{MS}_{*}$ be any $m$-variate quantum mapping. Further, let~$S_0,S_1 \subseteq \{0,1\}^n$ be two disjoint sets,~$d$ be a positive integer,~$\varepsilon > 0 $ be real, and~$s := \lceil n \ln 2 / (2\varepsilon^2) \rceil $. 

For any choice of positive real~$\ell$, if~$\red$ is~$\ell$-lossy for all~$(ds)$-uniform distributions that split over the pair $(S_0,S_1)$, then there exist two collections~$K_1,\ldots,K_s$ and~$T_1,\ldots,T_s$ of multisets of $d$ elements respectively contained in~$S_0$ and~$S_1$, such that
\begin{enumerate}
\item for any~$y \in S_0$, it holds that
\begin{align*}
\underset{\substack{a\sim \unif_{[s]}\\ \pi \sim \perm_{m} }}{\mathbb{E}} \left[  \Big\|
\red \left( \pi \left(
\unif_{K_a}^{\otimes m_0},y,\unif_{T_a}^{\otimes m_1} \right) \right) - \red\left( \pi \left( \unif_{K_a}^{\otimes (m_0+1)}, \unif_{T_a}^{\otimes m_1}\right) \right) 
 \Big\| \right] \\ 
 \leq \delta + \frac{2(m+1)}{d+1} + 2\varepsilon \ ,
\end{align*}
\item and for any~$y \in S_1$, it holds that
\begin{align*}
\underset{\substack{a\sim \unif_{[s]}\\ \pi \sim \perm_{m} }}{\mathbb{E}} \left[  \Big\|
\red \left( \pi \left(
\unif_{K_a}^{\otimes m_0},y,\unif_{T_a}^{\otimes m_1} \right) \right) - \red\left( \pi \left( \unif_{K_a}^{\otimes m_0}, \unif_{T_a}^{\otimes (m_1+1)}\right) \right) 
 \Big\| \right] \\ 
 \leq \delta + \frac{2(m+1)}{d+1} + 2\varepsilon \ ,
\end{align*}
\end{enumerate} 
where 
\begin{align*}
%		s\eqdef \left\lceil \frac{n \ln 2}{2\varepsilon^2} \right\rceil \ \text{and} \ 
		\delta \eqdef \min \left\{ \sqrt{\frac{\ell \ln 2}{2m}}, 1- 2^{-\frac{\ell}{m}-2} \right\}.
	\end{align*}
\end{lemma}

Note that the states inside the trace distance are mixed states since the inputs of~$\red$ are randomized classical distributions.

\begin{remark}
If the mapping~$R$ is $1$-variate and lossy only on~$S_0$-supported (or~$S_1$-supported) distributions, then only Item~1 (or Item~2) holds.  
\end{remark}

\begin{proof}[Proof of Theorem~\ref{th:spl-in-re}]
We assume that~$f$ is non-constant since otherwise the inclusion is trivial.
Let~$\Pi \in \SPL_f$ or~$\Pi \in \QSPL_f$ and~$R$ be its $\cute$ reduction.
\begin{algorithm}[h!]
	\caption{Randomized Encoding for~$\Pi$.}\label{algo:F} 
	\begin{algorithmic}[1]
		\item[\bf Parameters:]   
		Let $$
		S_0 \eqdef \Pi_N \cap \{0,1\}^n, \ \ S_1 \eqdef \Pi_Y \cap \{0,1\}^n, \ \ \varepsilon \eqdef \frac{1}{4q}, \ \
		d\eqdef \left\lceil \frac{m+1}{\varepsilon} \right\rceil, \  s\eqdef \left\lceil \frac{n \ln 2}{2\varepsilon^2} \right\rceil \ ,$$
		and $K_1,\ldots K_s,T_1,\ldots,T_s$ as in Lemma~\ref{lemma:perm-dis-dis}. $\varepsilon$ is set such that~$B=ds$ where $B$ is defined as in Definition~\ref{def:cute}.
		\item[\bf Input:] An instance~$y \in \{0,1\}^n$. 
		\item[\bf Advice:] 
		$K_a,T_a,\pi$ for some uniformly chosen $a\in[s]$ and $\pi\in\perm_m$. 
		\item[\bf Output:] A pair of circuits~$(\Sim,\Enc)$.

		\vspace{0.5cm}
		\State Let~$p(f)$ be as in Lemma~\ref{lemma:minimum-p}.
		\State Let~$\Enc(y)$ be the following circuit:
		it samples~$\widetilde{x} \sim \left(\unif_{K_a}^{\otimes m-p},y,\unif_{T_a}^{\otimes p-1}\right)$, then it outputs~$R(\pi(\widetilde{x}))$.
		\State Let~$\Sim(0)$ be the following circuit: 
		it samples~$\widetilde{x} \sim \left( \unif_{K_a}^{\otimes m-p+1}, \unif_{T_a}^{\otimes p-1} \right)$, then it outputs~$R(\pi(\widetilde{x}))$.
		\State Let~$\Sim(1)$ be the following circuit: 
		it samples~$\widetilde{x} \sim \left( \unif_{K_a}^{\otimes m-p}, \unif_{T_a}^{\otimes p} \right)$, then it outputs~$R(\pi(\widetilde{x}))$.
	\end{algorithmic} 
\end{algorithm}
We show that~$(\Enc,\Sim)$ defined as in Algorithm~\ref{algo:F} is a randomized encoding for~$\Pi$.
Before that, let us explain the input-output type of these circuits. Consider~$X \in \{\Sim,\Enc\}$. When~$\red$ is a randomized reduction, the two circuits are also randomized. Part of their randomness input is used to sample~$\widetilde{x}$ and the other part is fed to~$\red$. Let~$\kappa$ be the size of the total randomness. For~$r\in \{0,1\}^\kappa$, we let~$X(r)$ denote the outcome of~$X$ given the randomness~$r$. On the other hand, when~$\red$ is quantum, the circuits will be mixed algorithms; classical randomness is required for sampling~$\widetilde{x}$. Let~$\kappa'$ be the size of total randomness.\footnote{Note that~$\kappa$ and~$\kappa'$ are possibly different depending on how much classical randomness~$\red$ requires.} For any~$r \in \{0,1\}^{\kappa'}$, we let the mixed outcome of~$X$ be~$X\ket{r,\vec{0}}$ where~$\ket{\vec{0}}$ is some appropriate-size ancilla, emphasizing its mixed classical-quantum nature. When it is not relevant, we drop the dependency on~$r$ for simplification.

The proof is exactly similar for both classical and quantum case. Recall that, we need to show correctness and privacy properties as required in Definition~\ref{def:class-re}.

	\medskip
	\noindent\textit{Correctness:}
	It suffices to show that the distance (per Definition~\ref{def:ell1}) of the outcomes~$\Enc(y)$ and~$\Enc(y')$ for any choice of~$y\in \Pi_Y$ and~$y' \in \Pi_N$ is at least~$1-2\mu(n)$.
	We use the fact that the distance is equal to the advantage of the best distinguisher in the following game:
	Sample a uniform coin~$b\sim U_{\{0,1\}}$. If~$b=0$ let~$z \leftarrow \Enc(y')\ket{r,\vec{0}}$ where~$r$ follows the uniform distribution. Otherwise let~$z \leftarrow \Enc(y)\ket{r,\vec{0}}$. 
	We drop the dependency on~$r$ for simplification. 
	On one hand, if~$z$ is computed by~$\Enc(y')$, we have that~$\widetilde{x}:= (x_1,\ldots,x_m) \sim \left( \unif_{K_a}^{\otimes m-p}, y', \unif_{T_a}^{\otimes p-1} \right)$ with $K_a\subseteq \Pi_N\cap\{0,1\}^n$ and $T_a\subseteq \Pi_Y\cap\{0,1\}^n$. Then, since $\widetilde{x}$ contains $p-1$ YES instances, by Lemma~\ref{lemma:minimum-p}, for any~$\pi \in \mathfrak{S}_m$, we have
	$$f( \pi( \chi_\Pi(x_1),\ldots,\chi_\Pi(x_m) )) = 0 \ .$$ 

	On the other hand, if~$z$ is computed by~$\Enc(y)$, we have that~$\widetilde{x}$ contains one more YES instance~$y\in \Pi_Y \cap \{0,1\}^n$, therefore, $$f( \pi( \chi_\Pi(x_1),\ldots,\chi_\Pi(x_m) )) = 1 \ .$$

	Thus by the correctness of the reduction (per Definition~\ref{def:red}), there exists an algorithm~$\mathcal{M}$ such that
	\begin{align*}
		&\|\Enc(y')\ket{\mathbf{0}}-\Enc(y)\ket{\mathbf{0}}\|\\
		&
		\geq \mathbb{E}_{i \sim \mathcal{U}_{[m]}} \left| \Pr[1 \leftarrow \mathcal{M}(\Enc(y')\ket{\mathbf{0}})] - \Pr[1 \leftarrow \mathcal{M}(\Enc(y)\ket{\mathbf{0}})] \right|\\
		&  \geq 1-2\mu(n) \, ,
	\end{align*}
	where we used the fact that revealing $\pi$ with the description of the circuits does not decrease the success probability of any distinguisher.

	\medskip
	\noindent\textit{Privacy:} Suppose~$y\in \Pi_N$. We consider a modification of the distinguishing game where the random variables~$a$ and~$\pi$ are also given to the distinguisher. 
	Revealing~$a,\pi$ along with~$z$ does not decrease the success probability of the distinguisher, thus we can bound the original distinguishing probability by the distinguishing probability of the new task. It holds that 
	\begin{align*}
		\|\Sim(0)\ket{\mathbf{0}}-&\Enc(y)\ket{\mathbf{0}}\| 
		\\
		&\leq \left\|  \red\left( \pi \left( \unif_{K_a}^{\otimes m-p+1}, \unif_{T_a}^{\otimes p-1}\right) \right)  - \red \left( \pi \left(
		\unif_{K_a}^{\otimes m-p},y,\unif_{T_a}^{\otimes p-1} \right) \right) \right\|,
	\end{align*}
	By taking the expectation over~$a$ and~$\pi$, we have
	\begin{align*}
		\| &\Sim(0)\ket{\mathbf{0}}-\Enc(y)\ket{\mathbf{0}}\|\\ 
		&\leq \mathbb{E}_{a \sim \unif_{[s]},\pi \sim \perm_m} \Big[  \Big\|  \red\left( \pi \left( \unif_{K_a}^{\otimes m-p+1}, \unif_{T_a}^{\otimes p-1}\right) \right)
		- \red \left( \pi \left( \unif_{K_a}^{\otimes m-p},y,\unif_{T_a}^{\otimes p-1} \right) \right) \Big\|  \Big] \ .
	\end{align*}
	Note that by assumption~$\red$ is $\lambda$-lossy for all $\poly(n)$-uniform distributions that split over~$(\Pi_Y,\Pi_N)$. Therefore, one can choose $B=\poly(n)$ such that~$B=ds$. It follows by Lemma~\ref{lemma:perm-dis-dis} that
	\begin{align*}
		\|\Sim(0)\ket{\mathbf{0}}-\Enc(y)\ket{\mathbf{0}}\| &\leq\delta+\frac{2(m-p+1)}{d+1}+2\varepsilon\leq\delta+\frac{1}{q} \, .
	\end{align*}
	Similarly, for every~$y\in \Pi_Y$ we have $\|\Sim(1)\ket{\mathbf{0}}-\Enc(y)\ket{\mathbf{0}}\| \leq \delta+1/q$. 
	This completes the proof.
\end{proof}

{
One can also prove the following special case.
\begin{theorem}\label{th:ospl-in-re}
For any polynomial function~$q(n)$, it holds that $\OSPL[\lambda,\mu] \subseteq \OSRE[\delta+1/q(n),\mu]/\poly$ and $\OQSPL[\lambda,\mu] \subseteq \OQRE[\delta+1/q(n),\mu]/\poly$ where $\delta \eqdef \min \{ \sqrt{(\lambda \ln 2)/2}, 1- 2^{-\lambda-2}\}$ and the advice size is a polynomial function of~$q$.
\end{theorem}
\begin{proof}
The proof is exactly similar to that of Theorem~\ref{th:spl-in-re} for the special case where~$f$ is the identity function. Note that since the reduction is one-sided $\cute$, then Disguising Lemma~\ref{lemma:perm-dis-dis} only guarantees Item~2 regarding~$S_1=\Pi_Y$. Therefore, we obtain a one-sided randomized encoding. 
\end{proof}
}

{

\subsection{Better Analysis for Small $\lambda$}

A different variant of Theorem~\ref{th:spl-in-re} and \ref{th:ospl-in-re} can be obtained for sufficiently small values of $\lambda$.

\begin{theorem}\label{th:simpler-spl-in-re}
For any permutation-invariant function~$f$, we have $$\SPL_f[\lambda,\mu] \subseteq \SRE[\sqrt{2\ln 2\cdot m \lambda},\mu]/\poly, \quad \text{and} \quad \QSPL_f[\lambda,\mu] \subseteq \QRE[\sqrt{2\ln 2\cdot m \lambda},\mu]/\poly~.$$
Moreover, it holds that $$\OSPL[\lambda,\mu] \subseteq \OSRE[\sqrt{2 \ln 2 \cdot \lambda},\mu]/\poly, \quad \text{and} \quad \OQSPL[\lambda,\mu] \subseteq \OQRE[\sqrt{2 \ln 2 \cdot \lambda},\mu]/\poly~.$$
\end{theorem}

This result does not depend on the disguising lemma nor the minimax theorem, but rely on a different information-theoretic argument. Its advantage is that when $\lambda(n)$ is negligible for a $\cute$ problem, then the privacy of its upper-bounds are also negligible. This is in contrast with Theorem~\ref{th:spl-in-re} and \ref{th:ospl-in-re} where the privacy is $\delta+1/q$, which remains inverse polynomial.

\begin{proof}
We prove the first inclusion, the others are similar. Let $\Pi \in \SPL_f[\lambda,\mu]$. Take as advice $x_1^* \in \Pi_Y \cap \{0,1\}^n$ and $x_0^* \in \Pi_N \cap \{0,1\}^n$ (in the one-sided variant, only $x_1^*$ suffices) and~$p(f)$ be as in Lemma~\ref{lemma:minimum-p}. The encoding and the simulator are as follows:
\begin{equation*}
\Enc(y) := R\left( x_1^*,\cdots,x_1^*, x_0^*,\cdots, x_0^*, y\right),
\end{equation*}
and
\begin{equation*}
	\forall b \in \{0,1\}: \Sim(b) := R\left( x_1^*,\cdots,x_1^*, x_0^*,\cdots, x_0^*, x_b^*\right) \, ,
\end{equation*}
where there are $p(f)-1$ number of $x_1^*$ and $m-p(f)$ number of $x_0^*$. There correctness follows immediately by construction. 

We prove privacy. Fix an input $y$ and let $b:= \chi_\Pi(y)$ (for one-sided, $b$ is always $1$). Consider the following random variable $X:=(X_1,\cdots,X_m)$. For all $i < m$, set $X_i$ to be equal to the $i$-th entry of $\Enc$ as above and set $X_m$ to be $y$ with probability $1/2$ and $x_b^*$ otherwise. Note that its distribution is splitting, pair-wise independent and $2$-uniform as per Definition~\ref{def:cute}. Since $\Pi \in \SPL_f[\lambda,\mu]$, we have that $I(X;R(X)) \leq m\lambda$. Note also that the distribution of $R(X)$ follows $\Enc(y)$ or $\Sim(b)$, each with probability $1/2$. Let $M$ be a random variable that outputs $\Enc(y)$ with probability $1/2$ and $\Sim(b)$ otherwise. We have
\begin{align*}
I(X;R(X)) &= I(X_m;R(X))  \quad \text{(since } X_1,\cdots,X_m \text{ are fixed)}\\
&= \mathbb{E}_{X_m}\left[ D_{\mathrm{KL}} \left( p_{R(X)|X_m} \| p_{R(X)}\right) \right] \\
&=\frac{1}{2} D_{\mathrm{KL}}(p_{\Enc(y)} \| p_M) + \frac{1}{2} D_{\mathrm{KL}} (p_{\Sim(b)} \| p_M) \quad \text{(Jensen-Shannon Divergence)}\\
&\geq  \frac{\|\Enc(y) - \Sim(b) \|^2}{2 \ln 2} \, \quad \text{(Pinsker's inequality)},
\end{align*} 
where we used the fact that $\|\Enc(y) - M\| = \|\Sim(b) - M\| = \|\Enc(y) - \Sim(b)\|/2$. The statement then follows by noting that $I(X;R(X))\leq m\lambda$. The proof in the quantum setting is exactly similar with the difference being the use of quantum relative entropy instead of KL divergence.
\end{proof}

}

\section{Randomized Encodings Are Sparsely Lossy}\label{sec:re-in-spl}

In this section we analyze the $\cuteness$ of randomized encodings. 
Main theorem of this section is the following:
\begin{theorem}\label{th:re-in-spl}
The following statement holds: $$\OSRE[1-1/\poly(n),1/2-1/\poly(n)] \subseteq \OSPL[\negl(n),\negl(n)]\, .$$ 
\end{theorem}

We use the following lemma. For avoiding confusion, we differentiate a random variable $X$ from its distribution (probability mass function) $p_X$, in this section. Second item can be found in~\cite{CT06,corr:sas15}.
\begin{lemma}\label{lemma:mutual-KL-Delta}
It hold that:
\begin{enumerate}
\item For any random variables~$X$ and~$Y$, we have $I(X,Y) = \mathbb{E}_Y[D_{\textrm{KL}}(p_{X|Y} \| p_X)]$.
\item For any two distributions~$p \ll q$, i.e. $p$ is absolutely continuous with respect to $q$, then~$D_{\textrm{KL}}(p\|q) \leq (\frac{\log e}{\alpha_q}) \cdot \|p-q\|^2$ where~$\alpha_q = \min_x \{q(x) > 0\}$.
\end{enumerate}
\end{lemma}

\begin{proof}[Proof of Theorem~\ref{th:re-in-spl}]
Let~$\Pi \in \OSRE[d=1-1/\poly(n),\mu=1/2-1/\poly(n)]$ with the one-sided randomized encoding~$\Enc$ and the simulator~$\Sim$ for the YES instances. 
Let~$k(n)\geq(1-d(n))^{-1}$ and~$\ell(n)\geq (1/2-\mu(n))^{-1}$ be polynomials. Let
$$(P,T) := \mathsf{XOR}(\Enc,\Sim,1^k)\, ,$$
where~$\mathsf{XOR}$ is as in Lemma~\ref{lemma:xor}. Then define~$(R,S) := (P^{\otimes \ell},T^{\otimes \ell})$.
	The proof consists of showing that $\red$ is a reduction for~$\Pi$ with negligible error as per Definition~\ref{def:red} which is also one-sided $\negl(n)$-$\cute$ as per Definition~\ref{def:cute}. 
	
	We first prove the lossiness. Let~$B :=\poly(n)$ and~$X$ be a random variable whose distribution is $B$-uniform over $n$-bit strings and splits over $(\Pi_Y,\Pi_N)$ 
	as per Definition~\ref{def:split-lossy}. We bound the quantity $ I(X;\red(X))$ from above.
	For any random variable $Z$ let~$p_Z$ denote its distribution (probability mass functions). We first rewrite the mutual information in terms of Kullback-Leibler divergence as in Lemma~\ref{lemma:mutual-KL-Delta}
	\begin{align}\label{eq:develop_mi}
		I(X;R(X)) = \mathbb{E}_{R(X)}[ D_{\textrm{KL}}\left( p_{X | R(X)} \ \| \ p_{X} \right)]
%		=\sum\limits_{y \in \supp(R)} p_{R(X)}(y) \cdot D_{KL}\left( p_{X | R(X)=y} \ \| \ p_{X} \right)~.
	\end{align}
	By Lemma~\ref{lemma:mutual-KL-Delta}, the KL divergence of two distributions decreases as their trace distance does, in particular
	\begin{align*}
		D_{\textrm{KL}}\left( p_{X | R(X)} \ \| \ p_{X} \right) &\leq \left(\frac{\log e}{\alpha_{p_X}} \right) \cdot \|p_{X|R(X)}- p_X\|^2
	\end{align*}	
	where $\alpha_{p_X} = \min\limits_x \{p_X(x)>0\}$. 
	Substituting above in Equation~\eqref{eq:develop_mi}, we obtain:
	\begin{equation}\label{eq:minfo-upper-bound}
		I(X;R(X)) \leq \left(\frac{\log e}{\alpha_{p_X}} \right) \cdot \mathbb{E}_{R(X)}\left[ \|p_{X|R(X)}- p_X\|^2 \right] \, .
%		\sum_{y \in \supp(R(X))}  p_{R(X)}(y) \cdot \|p_{X|R(X)=y}-p_X \|_1^2  \, .
%		\max\{ 1, \ 2 - \log(\alpha_X) + 2\sum_{y \in \supp(R)} p_{R(X)}(y) \cdot \log(\Delta(X_{ | R(X)=y},X))   \}~.
	\end{equation}
	
	\noindent We split the bound on the right-hand side of the Inequality~\eqref{eq:minfo-upper-bound} into two terms.\\
	
	\noindent{\textbf{Bounding $\mathsf{term}_1 = -\log(\alpha_{X})$:}} Since $X_i$'s are pairwise independent $B$-uniform distributions, we have $\alpha_X = 1/B$.

	\noindent{\textbf{Bounding $ \mathsf{term}_2 = \mathbb{E}_{R(X)} [ \|p_{X|R(X)}- p_X\|^2 ]  $:}} Firstly, we have
	\begin{align*}
		\|p_{X|R(X)}-p_X \| &= \frac{1}{2} \sum\limits_{x,y} \Pr(R(X)=y) \cdot \left|\Pr \left(X=x | R(X) = y\right) - \Pr(X=x)\right| \\
		&= \frac{1}{2} \sum\limits_{x,y} \Pr(R(X)=y) \cdot \left| \frac{\Pr\left(X=x \land R(X) = y\right)}{\Pr(R(X)=y)} - \Pr(X=x) \right| \\
		&= \frac{1}{2} \sum\limits_{x,y} \left| \Pr\left(X=x \land R(X) = y\right) - \Pr(X=x)\cdot \Pr(R(X)=y) \right|\\
		&= \|p_{(X,R(X))}-p_X \otimes p_{R(X)}\| \, . 
	\end{align*}
	With a similar argument, we obtain that~$\|p_{X|R(X)}-p_X \| = \|p_{R(X)|X}-p_{R(X)} \|$. It follows that
	\begin{align*}
	\|p_{X|R(X)}-p_X \| &= \|p_{R(X)|X}-p_{R(X)} \|\\
	&= \sum\limits_{x} \Pr(X=x) \cdot \|p_{R(x)}-p_{R(X)}\|\\
	&\leq \max_x \|p_{R(x)}-p_{R(X)}\| \, .
	\end{align*}
	Recall that $(\red,S)$ is obtained by applying the XOR and then direct product to~$\Enc$ and~$\Sim$. 
	By the properties of the randomized encoding, for any $x \in \Pi_Y \cap \{0,1\}^n$, it holds that 
	$$\|p_{R(x)}-p_{S(x)}\| \leq \ell \cdot d^k \leq \negl(n) \, ,$$
	where we used Lemma~\ref{lemma:xor} and~\ref{lemma:direct-product}.
	Therefore $\|p_{R(X)}-p_{S(x)}\| \leq \negl(n)$ for any distribution of $X$ over $\Pi_Y \cap \{0,1\}^n$. We conclude that for any $x\in\Pi_Y \cap \{0,1\}^n$ and any distribution of $X$ over $\Pi_Y \cap \{0,1\}^n$, it holds that $\|p_{R(x)}-p_{R(X)}\|_1 \leq \negl(n)$. Therefore, we have
	\begin{align*}
		\mathbb{E}_{R(X)} \left[  \|p_{X|R(X)}- p_X\|^2 \right]  &\leq  \mathbb{E}_{R(X)}\left[ \|p_{X|R(X)}- p_X\|\right]^2 \quad (\text{Jensen's inequality}) \\
		&\leq  \mathbb{E}_{R(X)}\left[ \max_x \|p_{R(x)},p_{R(X)}\| \right]^2 \\
		&\leq \negl(n) \, . 
	\end{align*}

	Combining upper bounds on $\mathsf{term}_1$ and $\mathsf{term}_2$, we obtain
    $I(X;\red(X)) \leq 4 \log(e) \cdot \negl(n) \cdot B $
	for any splitting $B$-uniform distribution~$X$ of YES instances. For every polynomial~$B$, it follows that~$4 \log(e) \cdot \negl(n) \cdot B =\negl(n)$ which proves the one-sided $\cuteness$.
	
	\medskip
	 It remains to note that~$R$ is a reduction as in Definition~\ref{def:red}. This is implied by the correctness of the randomized encoding as follows.	 
	 First, note that by correctness, for every~$x\in \Pi_N$ and~$y\in \Pi_Y$, we have that~$\|p_{\Enc(x)} - p_{\Enc(y)}\| \geq 1-2\mu(n)$. Therefore, for every~$x\in \Pi_N$ and~$y\in \Pi_Y$, Lemma~\ref{lemma:xor-trick} implies that 
	 $\|p_{P(x)} - p_{P(y)}\| \geq 1/2-\mu(n)$.
	 Now, by using Lemma~\ref{lemma:direct-product}, we have
	 \begin{align*}
	 \|p_{R(x)} - p_{R(y)}\| &= \| p_{P(x)}^{\otimes \ell} - p_{P(y)}^{\otimes \ell} \| \\
	 &\geq 1- 2\exp\left( - \frac{\ell  \cdot \| p_{P(x)} - p_{P(y)}\|^2}{2} \right) \\
%	 &\geq 1- 2\exp\left( - \frac{\ell  \cdot \| p_{\Enc(x)} - p_{\Enc(y)}\|^2}{2} \right) \\
	 &\geq 1-\negl(n) \, .
	 \end{align*}
%	 where we used the direct product lemma. 
	 Therefore, for every~$x\in \Pi_N$ and~$y\in \Pi_Y$, there exists a distinguisher~$\mathcal{M}_{x,y}$ that tells~$R(x)$ apart from~$R(y)$ with probability~$1-\negl(n)$. By Lemma~\ref{lemma:minmax}, there exists an~$\mathcal{M}$ that for every~$x\in \Pi_N$ and~$y\in \Pi_Y$ distinguishes~$R(x)$ from~$R(y)$ with probability~$1-\negl(n)$ (note that the quantifiers are swapped). This~$\mathcal{M}$ satisfies the requirement in Lemma~\ref{def:red}, which completes the proof.
	
\end{proof}

One might hope to obtain a similar statement for two-sided randomized encodings and~$\SPL$. However, the above technique cannot be used. The barrier comes from the fact that the reduction is XOR of the encoding~$\Enc(\cdot)$ and the simulator~$\Sim(\cdot)$ for the YES instances. In the two-sided setting, in order to keep the reduction lossy, we would have to use the correct simulator for the instance, i.e.,~$\Sim(\chi_\Pi(\cdot))$. Then the reduction could possibly go beyond the polynomial time. On the positive side, we can reuse the above argument when the distance~$d$ and error~$\mu$ are both negligible. In this setting, there is no need to amplify these quantities via XORing and repeating. The reduction is in fact~$\Enc$ itself. 
\begin{theorem}
It holds that~$\SRE[\negl(n),\mu] \subseteq \SPL[\negl(n),\mu]$.
\end{theorem}

We obtain a quantum variant of Theorem~\ref{th:re-in-spl}. The construction of the randomized encoding is the same but the proof is slightly different. This is because lack of a reverse-Pinsker-like inequality for the quantum states. Instead, we use Alicki–Fannes–Winter Inequality~\ref{thm:afw_inequality} which gives an upper bound for conditional Shannon entropy in terms of the trace distance. 

\begin{theorem}\label{thm:qre-qspl}
The following statement holds: 
$$\OQRE[1-1/\poly(n),1/2-1/\poly(n)] \subseteq \OQSPL[\negl(n),\negl(n)]\, .$$ 
\end{theorem}

We first recall the required preliminaries.
We define the quantum conditional entropy of a two-system state $\rho_{AB}$ as follows
\begin{align*}
	S(A|B) := S(\rho_{AB}) - S(\rho_B)\, ,
\end{align*}
The quantum mutual information in terms of conditional quantum entropy is
$$ I(A;B)_\rho = S(\rho_A) - S(A|B)_\rho = S(\rho_B) - S(B|A)_\rho\, .$$

\begin{lemma} Let $\rho_{AB}$ be a quantum state in two subsystems $A$ and $B$, with marginal states~$\rho_A$ and~$\rho_B$. The following properties hold.
	\begin{enumerate}
		\item The Von Neumann entropy is additive for tensor product states: $ S(\rho_A\otimes\rho_B) = S(\rho_A) + S(\rho_B)\, .$
		\item Conditioning does not increase entropy: $ S(\rho_A)\geq |S(A|B)_\rho|\, .$
		\item The quantum entropy of a system is bounded by the dimension: $ S(\rho_A)\leq\dim(H_A)\, . $
	\end{enumerate}
\end{lemma}

\begin{lemma}[Alicki–Fannes–Winter Inequality~\cite{wilde}]\label{thm:afw_inequality}
	Given two bipartite quantum states $\rho_{AB},\omega_{AB}\in\mathcal{D}(H_A\otimes H_B)$, then
	$$\left| S(A|B)_\rho-S(A|B)_\omega \right|\leq 2\|\rho-\sigma\|\log\dim(H_A)+h(\|\rho-\sigma\|)\, , $$
	where~$h(p):= -p \log p -(1-p) \log (1-p)$ is the binary entropy function.
\end{lemma}

\begin{proof}[Proof of Theorem~\ref{thm:qre-qspl}]
	Let~$\Pi \in \OQRE[d=1-1/\poly(n),\mu=1/2-1/\poly(n)]$ with the one-sided quantum randomized encoding~$\Enc$ and the simulator~$\Sim$ for the YES instances. 
Let~$k(n)\geq(1-d(n))^{-1}$ and~$\ell(n)\geq (1/2-\mu(n))^{-1}$ be polynomials. Let
$$(P,T) := \mathsf{XOR}(\Enc,\Sim,1^k)\, ,$$
where~$\mathsf{XOR}$ is as in Lemma~\ref{lemma:xor}. Then define~$(R,S) := (P^{\otimes \ell},T^{\otimes \ell})$.
	We show that $\red$ is a quantum reduction for~$\Pi$ that is one-sided $\negl(n)$-$\cute$ as per Definition~\ref{def:cute} and has negligible error as per Definition~\ref{def:red}. 
	
	Let us denote by~$\rho_{X,\red(X)}$ (or simply by~$\rho$) the joint system of the classical-quantum state after the reduction~$\red$ is applied to distribution~$X$, see Equation~\eqref{eq:c-q_map}. We denote the subsystems of~$\rho$ by~$\rho_X$ and~$\rho_{\red(X)}$. Note that since~$\rho$ is a classical-quantum system, so is~$\rho_X\otimes\rho_{\red(X)}$. Let also~$\omega$ be the product state~$\omega_{X,\red(X)}:=\rho_X\otimes\rho_{\red(X)}$, since the mutual information between subsystems of product states are zero, we have
	\begin{align*}
		I(X;\red(X))_\rho & = |I(X;\red(X))_\rho-I(X;\red(X))_\omega|\\
		& = |S(\rho_X) - S(X|\red(X))_\rho -S(\omega_X) + S(X|\red(X))_\omega| \\ 
		& = |S(X|\red(X))_\rho-S(X|\red(X))_\omega |\\
		&\leq 2\|\rho-\omega\|\log\dim(\operatorname{Im}(\red(X)))+h(\|\rho-\omega\|)\, ,
	\end{align*}
	where in the last inequality we used Lemma~\ref{thm:afw_inequality}. We have
	\begin{align*}
		\|\rho-\omega\|=\|\rho_{X,\red(X)}-\rho_X\otimes\rho_{\red(X)}\| & = \sum_x \Pr(X=x)\|\rho_{\red(x)}-\rho_{\red(X)}\| \, .\\
%		& \leq \sum_x \Pr(X=x) 2d = 2d\, .
	\end{align*}
	By following the exact same argument as in the proof of Theorem~\ref{th:re-in-spl}, we obtain that~$\|\rho_{\red(x)}-\rho_{\red(X)}\| \leq 2 \ell d^k$ for every~$x \in \Pi_Y$. It then follows that 
	\begin{align*}
	I(X;\red(X))_\rho &\leq 2\|\rho-\omega\|\log\dim(\operatorname{Im}(\red(X)))+h(\|\rho-\omega\|) \\
	&\leq 4\ell \cdot d^k \log\dim(\operatorname{Im}(\red(X))) + h(2\ell \cdot d^k)\\
	&\leq \negl(n)\, ,
	\end{align*}
	where we used the fact that $\dim(\operatorname{Im}(\red(X))) = \poly(n)$ and $h$ is an increasing function on inputs~$<1/2$. This completes the privacy part. 
	
	The correctness is similar to the proof of Theorem~\ref{th:re-in-spl}.
	
\end{proof}

\section{Randomized Encodings and One-Wayness}

In this section, we investigate the connections between one-sided randomized encodings, OWFs and OWSGs. The first is about the existence of OWFs.

\begin{theorem}\label{th:sre-owf}
If~$\OSRE[1-1/\poly(n),0]\not\subseteq \mathsf{P}/\poly$ then OWFs exist.
\end{theorem}

\begin{proof}
We prove OWFs exist if~$\OSRE[1/2-1/\poly(n),0] \not\subseteq \BPP$. Then by Theorem~\ref{th:ospl-in-re} and~\ref{th:re-in-spl}, we have that~$\OSRE[1/2-1/\poly(n),0] = \OSRE[1-1/\poly(n),0]$ (up to some advice). Moreover, it holds that~$ \BPP/\poly \subseteq \mathsf{P}/\poly$ via standard repeat-then-majority derandomization. These two then yield the theorem.

{
Let~$\Pi$ be a problem in~$\OSRE[d:=1/2-1/n^c,0] \setminus \BPP$ for some constant~$c$. Let~$\Enc$ and~$\Sim$ be its one-sided encoding and simulator.
Let~$T(n)=\poly(n)$ be an upper bound on the runtime of~$\Enc$ and~$\Sim$. Without loss of generality, we can assume that the size of the randomness of these circuits is~$\kappa(n) \leq T(n)$ bits in both cases.\footnote{It is not difficult to show that if $\kappa=O(\log n)$ then $\Pi\in \BPP$.}
Let~$x$ be an arbitrary YES instance of $\Pi$ and~$(C_0(\cdot),C_1(\cdot))$ be $(\Sim(\cdot),\Enc(x;\cdot))$ that take $r\in \{0,1\}^\kappa$ as input.
	We show that~$\mathsf{F}:\{0,1\}^\kappa \times \{0,1\} \rightarrow \{0,1\}^*$, defined as $\mathsf{F}(r,b):= C_b(r)$ is a weak OWF. 
	 
Assume that there exists an adversary~$\mathcal{A}$ that breaks the scheme above with probability more than~$1-1/(2n^c)$. We use~$\mathcal{A}$ to construct an algorithm for~$\Pi$. Consider the following oracle algorithm~$\mathcal{B}^\mathcal{A}$:
\begin{itemize}
	\item $\mathcal{B}^\mathcal{A}(y)$: 
	Let~$\widehat{C}_0(\cdot):=\Sim(\cdot)$ and $\widehat{C}_1(\cdot):=\Enc(y;\cdot)$. 
Sample a uniform~$r\in \{0,1\}^\kappa$ and a uniform~$b \in \{0,1\}$, and let~$z := \widehat{C}_b(r)$.
	 Run the adversary~$(r',b') \leftarrow \mathcal{A}(z)$, and compute~$z' = C_{b'}(r)$. If $z=z' $ output~$1$, otherwise output~$0$. 
 
\end{itemize}

$\mathcal{B}$ mimics the behavior of~$\mathsf{F}$ but it replaces~$(C_0,C_1):=(\Sim,\Enc(x))$ with~$(\widehat{C}_0,\widehat{C}_1)=(\Sim,\Enc(y))$ in the state generation algorithm. 
We compute the advantage of $\mathcal{B}$ in distinguishing between YES and NO instances of $\Pi$ by analyzing the probability~$\Pr ( \mathcal{B}^\mathcal{A}(y) = 1 ) $.

\smallskip
	\noindent\underline{Case $y\in\Pi_N$}:
	We show that for every~$y \in \Pi_N$, we have: 
\begin{align*}
%\big|\Pr_{r}(\mathcal{B}(\widehat{C}_0(r))=1) - \Pr_{r}(\mathcal{B}(\widehat{C}_1(r))=1)\big| 
\Pr(\mathcal{B}^\mathcal{A}(y) = 1) \leq \frac{1}{2} \ .
\end{align*}
If~$b=0$, then~$\mathcal{A}$ succeeds to find an inverse with probability at most~$1$. Therefore, $\mathcal{B}$ outputs~$1$ with probability at most~$1$. On the other hand, if~$b=1$ then the adversary is fed with~$\Enc(y)$. By the correctness property, we know that the distance between~$C_1=\Enc(x)$ and~$\widehat{C}_1=\Enc(y)$ is~$1$. 
This means that the images of~$C_1$ and~$\widehat{C}_1$ are distinct. Therefore, there exists no~$r'$ such that $\widehat{C}_1(r) = C_1(r')$. So the adversary never succeeds and $\mathcal{B}$ never outputs~$1$ in this case. This yields the above upper bound. 

\smallskip
	\noindent\underline{Case $y\in\Pi_Y$}:
	By the privacy property, the distance of the outcomes of $C_1 = \Enc(x)$ and $\widehat{C}_1=\Enc(y)$ is at most~$2d$ since both of them are YES instances. Moreover,~$\widehat{C}_0$ is exactly the same as~$C_0$. Therefore, if the bit~$b$, chosen by~$\mathcal{B}$ is equal to~$0$, then~$\mathcal{A}$ succeeds with probability at least~$1-1/(2n^c)$, and if~$b=1$, it succeeds with probability at least~$1-1/(2n^c)- 2d$. In total, we obtain
	\begin{align*}
	\Pr(\mathcal{B}^\mathcal{A}(y) = 1) \geq \frac{1}{2}(1-\frac{1}{2n^c}) + \frac{1}{2}(1-\frac{1}{2n^c}-2d)= 1-\frac{1}{2n^c} - d \geq \frac{1}{2} + \frac{1}{2n^c}\, .
	\end{align*}
}

\smallskip
	\noindent\underline{Reduction}: The quantity of $\Pr(\mathcal{B}^\mathcal{A}(y) = 1)$ is different for YES and NO instances of~$y$ with an inverse polynomial gap~$\gap:=1/(2n^c)$. To decide~$\Pi$ it suffices to run~$\mathcal{B}$ for~$\poly(1/\gap)$ many times, and approximate the quantity above within error less than~$\gap/2$. If this value is more than~$1/2+\gap/2$, then~$y$ must be a YES instance, otherwise it is a NO instance. We precise the idea in the following.
%Let~$\thetaOWF:=1/2-3/(4n^c)$.
	Let~$k\in\mathbb{N}$. On instance $y\in\{0,1\}^n$, independently run~$\mathcal{B}(y)$ for~$k$ times and check whether its output correctly guesses the internal random bit~$b$ or not. Let~$b_1,\ldots,b_k$ be~$k$ corresponding bits representing whether the checks were passed or not. By definition, the average of~$b_1,\cdots , b_k$ must be an estimation of the success probability of~$\mathcal{B}$ in guessing the internal random bit~$b$. In other words, it approximates the quantity~$1/2+ \gap/2$. Then return as follows:
	\begin{equation*}
		\begin{cases}
			0 & \text{if } \left|\frac{1}{k}\sum_i b_i-\frac{1}{2}+\frac{1}{4n^c}\right|\geq \frac{\gap}{2},\\
			1 & \text{otherwise}\, .
		\end{cases}
	\end{equation*}
	 To analyze the success probability of the algorithm above in deciding~$y$, we proceed as follows. Let~$\mathcal{C}$ be the algorithm. We have
	\begin{equation*}
		\begin{split}
			&\Pr(\mathcal{C}(y)=0 | y\in\Pi_Y)\\
			&= \Pr\left(\left|\frac{1}{k}\sum_i b_i-\frac{1}{2}+\frac{1}{4n^c}\right|\geq \frac{\gap}{2} \ \Big| \ b_1,\ldots,b_k,y\in\Pi_Y\right)\\
			&\geq\Pr\left(\frac{1}{k}\sum_i b_i\geq \frac{1}{2}- \frac{1}{4n^c} \frac{\gap}{2} \ \Big| \ b_1,\ldots,b_k,y\in\Pi_Y\right)\\
			&\geq \Pr\left(\frac{1}{k}\left(\sum_i b_i- \mathbb{E}(b_i)\right)\geq -\frac{\gap}{2}\ \Big| \ b_1,\ldots,b_k,y\in\Pi_Y\right)\\
			&\geq 1-\exp\left(-\frac{k \cdot \gap^2}{2}\right) \ ,
		\end{split}
	\end{equation*}
	where we used $\mathbb{E}(b_i)-\gap/2\geq 1/2 -1/(4n^c) - \gap/2$ for $y\in\Pi_Y$ as shown earlier for the second inequality, and used Hoeffding's lemma in the last inequality. On the other hand, we have
	\begin{equation*}
		\begin{split}
			&\Pr(\mathcal{C}(y)=1|y\in\Pi_N)\\
			&=\Pr\left(\left|\frac{1}{k}\sum_i b_i-\frac{1}{2} +\frac{1}{4n^c}\right|< \frac{\gap}{2} \ \Big| \ b_1,\ldots,b_k,y\in\Pi_N\right)\\
			&= \Pr\left(\left|\frac{1}{k}\sum_i b_i-\frac{1}{k}\sum_i \mathbb{E}(b_i)\right|< \frac{\gap}{2} \ \Big| \ b_1,\ldots,b_k,y\in\Pi_N\right)\\
			&\geq 1-\exp\left(-\frac{k \cdot \gap^2}{2}\right) \ .
		\end{split}
	\end{equation*}
	For~$k := 4/\gap^2$, any sufficiently large~$n \in \mathbb{N}$, and any~$y \in (\Pi_Y \cup \Pi_N) \cap \{0,1\}^n$, it holds that
	\begin{equation*}
		\Pr(\mathcal{C}(y)=\chi_\Pi(y)) \geq 1-\exp\left(-\frac{k \cdot \gap^2}{2}\right)\geq\frac{2}{3} \ ,
	\end{equation*}
	Therefore,~$\mathcal{C}$ solves~$\Pi$ which is a contradiction. Therefore, no polynomial-time adversary~$\mathcal{A}$ has advantage larger than~$1-1/(2n^c)$ in inverting~$\mathsf{F}$.

\end{proof}

As mentioned in the introduction, building OWFs from worst-case hardness of~$\SZK$ is of great interest. The above theorem proposes a possible approach for doing so. The missing puzzle is a procedure that allows to completely remove the error from one-side randomized encodings even it means to increase the privacy. 
We prove the following theorem.
\begin{theorem}\label{th:owf-th}
If~$\OSRE[\negl(n),\negl(n)] \subseteq \OSRE\left[1-1/\poly(n),0\right]/\poly$ and $\SZK$ is worst-case hard, then OWFs exist.
\end{theorem}
\begin{proof}
First note that if $\SZK$ is worst-case hard then~$\NISZK$ is worst-case hard by Lemma~\ref{lemma:gsv99}. We also have~$\NISZK \subseteq \OSRE[\negl(n),\negl(n)]/\poly$ by Lemma~\ref{lemma:re-is-szk}. The rest follows from Theorem~\ref{th:sre-owf} and the assumption.
\end{proof}

Next theorem is the adaptation of Theorem~\ref{th:sre-owf} for OWSGs. The construction is similar. We let the generation algorithm of the OWSG to be the quantum mapping $\sgen$ that maps~$(b,r)$ to~$\Sim(r)$ if~$b=0$ and~$\Enc(x;r)$ if $b=1$, where~$x$ is an arbitrary YES instance. The verification algorithm takes~$(r,b)$ and a quantum state~$\rho$ as inputs and accepts if~$\|\sgen(r,b) - \rho\| \leq 1/3$. 

\begin{theorem}\label{th:owsg}
If~$\OQRE_{\mathsf{pure}}[1-1/\poly(n),0] \not\subseteq \BQP/\poly$ then OWSGs exist.
\end{theorem}
\begin{proof}
It suffices to prove the theorem for $\OQRE_{\mathsf{pure}}[1/2-1/\poly(n),0] \not\subseteq \BQP$ since by Theorem~\ref{th:ospl-in-re} and~\ref{th:re-in-spl}, we have that~$\OQRE_{\mathsf{pure}}[1/2-1/\poly(n),0] = \OQRE_{\mathsf{pure}}[1-1/\poly(n),0]$ (up to some polynomial-size classical advice).\footnote{The statements of Theorem~\ref{th:ospl-in-re} and~\ref{th:re-in-spl} does not explicitly address $\OQRE_{\mathsf{pure}}$, however the transformations that are involved in the proofs respect the purity of the encodings and reductions.}
Let~$\Pi \in \OQRE_{\mathsf{pure}}[d=1/2-1/n^c,0]\setminus \BQP$ for some constant~$c$ with one-sided encoding~$\Enc$ and simulator~$\Sim$ for YES instances of~$\Pi$. 
Let~$x$ be an arbitrary YES instance of $\Pi$. Let~$(C_0,C_1)$ be $(\Sim,\Enc(x))$.
	Note that the two circuits are mixed; a classical randomness is used.
	Let~$\kappa$ be the size of the randomness of these circuits.\footnote{It is not difficult to show that if $\kappa=O(\log n)$ then $\Pi\in \BQP$.}
	We show that~$\mathsf{G}$, defined as follows is a weak OWSG:
	\begin{itemize}
		\item $\sgen(r,b):$ output~$C_b\ket{r,\vec{0}}$.
		\item $\ver((r,b),\rho):$ If~$\|C_b\ket{r,\vec{0}}-\rho\|_1 \leq 1/3$ output~$1$, otherwise output~$0$.
	\end{itemize}
	This suffices to build OWSG according to~\cite{MY24}.
	 
Assume that there exists an adversary~$\mathcal{A}$ that breaks the scheme above with probability more than~$1-1/(2n^c)$ and let $\mathcal{B}^\mathcal{A}$ be the following oracle algorithm:
\begin{itemize}
	\item $\mathcal{B}^\mathcal{A}(y)$: 
	\begin{itemize}
	\item State Generation Phase: Let~~$\widehat{C}_0(\cdot):=\Sim(\cdot)$ and $\widehat{C}_1(\cdot):=\Enc(y;\cdot)$. 
Sample a uniform~$r\in \{0,1\}^\kappa$ and a uniform~$b \in \{0,1\}$, and let~$\rho := \widehat{C}_b\ket{r,\vec{0}}$.
	\item Test Phase:  Run the adversary~$(r',b') \leftarrow \mathcal{A}(\rho^{\otimes n})$, and compute~$\rho' = C_{b'}\ket{r',\vec{0}}$. If $\|\rho-\rho'\|\leq 1/3 $ output~$1$, otherwise output~$0$. 
    \end{itemize}		
\end{itemize}

Similar as in Theorem~\ref{th:sre-owf}, one can show that $ \Pr(\mathcal{B}^\mathcal{A}(y) = 1) \leq 1/2$ if $y \in \Pi_N$, and $\Pr(\mathcal{B}^\mathcal{A}(y) = 1) \geq 1/2 + 1/(2n^c) $ if $y \in \Pi_Y$.
 To decide~$\Pi$ it suffices to run~$\mathcal{B}$ for~$\poly(n^c)$ many times, and approximate the quantity above within error less than~$1/(4n^c)$. If this value is more than~$1/2+1/(4n^c)$, then~$y$ must be a YES instance, otherwise it is a NO instance. Therefore, we finally obtain a algorithm that solves~$\Pi$. 
 Note that~$\mathcal{B}$ tests whether~$\|\widehat{C}_b\ket{r,\vec{0}}-\widehat{C}_{b'}\ket{r',\vec{0}}\|_1$ is smaller than~$1/3$. 
 Since the encoding~$\Enc$ is pure and~$r,r'$ are fixed, these states are pure, therefore~$\mathcal{B}$ can perform a SWAP test for $\poly(n)$ number of times on them to approximate their distance. 
 \end{proof}

\section{Obfuscation}\label{sec:io}

We start this section by defining the obfuscation functionality as below.
\begin{definition}[Statistical Obfuscation (sO),~\cite{GR07}]\label{def:so}
Let~$N$ be a positive integer, and~$\alpha,\varepsilon: \mathbb{N} \rightarrow [0,1)$ be functions. An~$\alpha$-statistical obfuscation with error~$\varepsilon$ is an algorithm~$\mathsf{sO}$ that, upon receiving a security parameter~$1^N$ and a circuit~$C$ of size~$\poly(N)$ as inputs, runs in time~$\poly(N)$ and outputs a circuit with the following specifications:
\begin{enumerate}
\item \textbf{Correctness.} For any circuit~$C$ over inputs of size~$N$, it holds that
$$\Pr[\exists x \in \{0,1\}^N \ : \ \mathsf{sO}(1^N,C)(x) \neq C(x)] \leq \varepsilon \, ,$$
where the probability is taken over the randomness of~$\mathsf{sO}$.
\item \textbf{Statistical Distance.} For all circuits~$C_1$ and~$C_2$ that are functionally equivalent over inputs of size~$N$ such that~$|C_1|=|C_2|$, it holds that
\begin{equation}\label{eq:io}
\Delta(\mathsf{sO}(1^N,C_1),\mathsf{sO}(1^N,C_2)) \leq \alpha(N)\, .
\end{equation}
\end{enumerate}
\end{definition}

\begin{remark}[Indistinguishability obfuscation]
An indistinguishability obfuscation (iO) is defined similarly where the error is $\negl(N)$ and Equation~\eqref{eq:io} is replaced by computational indistinguishability: No polynomial-time adversary can distinguish between the two obfuscations with probability more than~$\negl(N)$. We say that iO is weak if the computational distinguishing probability is bounded by~$1-1/\poly(N)$.
\end{remark}

The main theorem of this section the following one.
\begin{theorem}\label{thm:io-nexist}
Assuming that Polynomial Hierarchy does not collapse to its third level,~$(1-1/\poly(N))$-statistical obfuscation with error~$1/2-1/\poly(N)$ does not exist. 

Assuming that refuting $\SAT$ is exponentially hard for non-uniform algorithms,~$(1-1/\subexp(N))$-statistical obfuscation with error $1/2-1/\subexp(N)$ does not exist.
\end{theorem}

The assumption that refuting $\SAT$ is exponentially hard is also known as non-deterministic exponential time hypothesis. It is is introduced by~\cite{CGI+16}, and more formally states that $\SAT$ does not have a non-uniform reduction to~$\mathsf{co}\NP$ that runs in time subexponential in~$N$, where~$N$ is the number of variables in the~$\SAT$ instance. It has been further studied by~\cite{CGI+16,CRTY23}.

\begin{remark}
We note that this impossibility result should not be confused with the positive result of Brakerski, Brzuska, and Fleischhacker~\cite{BBF16}. They construct an efficient obfuscator when~$2\varepsilon + \alpha > 1$~\cite[Appendix A]{BBF16}, but their correctness definition of obfuscation is quite different. They require that the obfuscator correctly evaluates the original circuit only over a uniformly random input, which is a relatively weaker guarantee.
\end{remark}

The following is an immediate application of the main theorem.
\begin{corollary}
If Polynomial Hierarchy does not collapse to its third level, then weak iO implies OWFs. If refuting $\SAT$ is exponentially hard for non-uniform algorithms, then weak iO with subexponential time implies OWFs.
\end{corollary}
\begin{proof}
Recall that iO is an obfuscation scheme with negligible error~$\varepsilon(N)=\negl(N)$ such that the obfuscation of any two functionally equivalent circuits are computationally indistinguishable. On the other hand, in Theorem~\ref{thm:io-nexist} implies that there are two family of circuits whose obfuscations are statistically far assuming that Polynomial Hierarchy does not collapse. Consider these obfuscations. They are computationally close but statistically far with a gap that is inverse polynomially large. This is known as EFI pairs which imply OWFs~\cite{Gol90,NR06,BDRV19}. The proof for the second part is similar and can be obtained by changing the parameter~$N$.
\end{proof}

In the rest of this section, we prove the main result. We need the following lemma.
We are interested in the unique satisfiability problem ($\uSAT$) which is similar to $\SAT$ with the extra promise that the input 3CNF formula has at most one satisfiable assignment. We also consider the regime of parameters where the runtime of the encoding and simulator can be superpolynomial since it has interesting consequences.

\begin{lemma}\label{lemma:io-wc-dist}
Assume that there exists an~$\alpha$-statistical obfuscation~$\mathsf{sO}$ with error~$\varepsilon$ such that $\alpha=1-1/\poly(N)$ and~$\varepsilon = 1/2 - 1/\poly(N)$.
Then there exists a one-sided randomized encoding for $\uSAT$ with privacy $1-1/\poly(n)$, error $1/2-1/\poly(n)$, and the same runtime as~$\mathsf{sO}$, where
$n$ is the size of the instance of $\uSAT$.
\end{lemma}
\begin{proof}
Let~$C$ be a YES instance of~$\uSAT$, {~i.e., it evaluates to 1 on exactly one input,} then a random shift of~$C$, namely the circuit~$C_z(x) = C(x \oplus z)$, where $z \rnd \{0,1\}^N$, has a truth-table whose distribution is identical to that of a random point function. We let the encoding to be this random shift followed by applying~$\mathsf{sO}$. The one-sided simulator is~$\mathsf{sO}(\mathds{1}_{x=z})$, where $\mathds{1}_{x=z}$ is
the point function on~$z$ where~$z$ is sampled uniformly at random from~$\{0,1\}^N$.
We analyze the distance $d$, error $\mu$ of the encoding. 
If $C$ is satisfiable, then the truth table of the point function $\mathds{1}_{x=z}$, where~$z \rnd \{0,1\}^N$, is statistically equivalent to that of~$C_z$. Therefore, by the statistical indistinguishability of~$\mathsf{sO}$, 
if~$C$ is satisfiable, then $\Delta(C_i, \mathsf{sO}(1^N,\mathds{1}_{x=z})) \leq \alpha(N)$. 
Note that the size of the instance~$n=|C|$ is~$\poly(N)$. Without loss of generality, we can assume that $|C| = \Omega(N)$; else, we can pad $C$ with enough garbage bits to add to its length, and throw out these extra bits in the reduction. Thus, if~$\alpha(N)=1-1/q(N)$ for some polynomial function~$q$ then the privacy is~$\leq 1-1/q(n) = 1-1/\poly(n)$.
The error follows immediately by the correctness of the obfuscation.
\end{proof}

The above lemma shows that $\uSAT$ reduces to~$\OSRE[1-1/\poly(n),1/2-1/\poly(n)]$ assuming that such an obfuscation exists. In the following, we prove that it is unlikely that $\uSAT$ has such a randomized encoding, and therefore, the assumption must be false.

Valiant and Vazirani~\cite{VV86} show that there exists a reduction from~$\SAT$ to~$\uSAT$ such that given an instance~$C$ (a circuit over inputs~$x \in \{0,1\}^N$), runs in time~$\poly(N)$ and has error at most~$1-1/\poly(N)$. The following lemma is obtained using their reduction.
\begin{lemma}\label{lemma:ext-vv}
There is a reduction from~$\SAT$ to~$\MAJ \circ \uSAT$ that runs in time~$\poly(N)$, has zero error and takes a non-uniform advice as input.
\end{lemma}
\begin{proof}
Let~$C$ be an instance of~$\SAT$ with size at most~$q(N)$. Let~$R$ be the Valiant-Vazirani reduction and~$p$ be such that~$R$ has error at most~$1-1/p(N)$. By repeating the reduction for a sufficiently large~$k= \poly(p,N)$ times and taking the majority, the error reduces to~$e^{-\Omega(k)}$ according to Lemma~\ref{lemma:hoeff}. Moreover, the union bound implies that~$R$ errs over at least one circuit~$C$ of size~$q(N)$ with probability at most $2^{q(N)} \cdot e^{-\Omega(k)}$. By carefully tuning~$k$ to be a polynomial larger than~$q$, this error quantity becomes strictly smaller than~$1$. Therefore, there exists at least one random string such that the reduction correctly maps all the instances~$C$ of size~$q$. We set this random string as the non-uniform advice. Therefore, this reduction runs in $\poly(N)$ time and has zero error. 
\end{proof}

Finally, we can prove the main theorem.
\begin{proof}[Proof of Theorem~\ref{thm:io-nexist}]
We prove the statements separately: 
\paragraph{Proving the first statement:} 
Assume~$(1-1/\poly(N))$-statistical obfuscation with error~$1/2-1/\poly(N)$ exists. Lemma~\ref{lemma:io-wc-dist} presents a one-sided randomized encoding for $\uSAT$ with privacy $1-1/\poly(n)$ and error $1/2-1/\poly(n)$. On the other side, we have that~$\SRE[1-1/\poly(n), 1/2-1/\poly(n)] \subseteq \OSRE[\negl(n),\negl(n)]$ by Theorem~\ref{th:ospl-in-re},~\ref{th:simpler-spl-in-re} and~\ref{th:re-in-spl} (or simply Key Lemma in Section~\ref{sec:contrib}). Then Lemma~\ref{lemma:re-is-szk} implies that~$\uSAT$ reduces to $\SZK$ in polynomial time. On the other hand, Lemma~\ref{lemma:ext-vv} provides a zero-error reduction from~$\SAT$ to~$\MAJ \circ \uSAT$. Combining these two reductions, we obtain a zero-error reduction from $\SAT$ to $\MAJ \circ \SZK$ that runs in polynomial time. The class~$\SZK$ is closed under Boolean formulas by Lemma~\ref{lemma:f(SZK)}, therefore, the whole chain of reductions imply that~$\SAT \in \SZK/\poly$. Moreover,  $\SZK/\poly \subseteq \mathsf{co}\NP/\poly$  by Lemma~\ref{prelim:lemma-szk-derand}. Therefore, we obtain~$\NP \subseteq \mathsf{co}\NP/\poly$ which contradicts the assumption according to Lemma~\ref{lemma:PH}.

\paragraph{Proving the second statement:}
The proof is exactly the same with the only difference being the runtime of the reduction from~$\SAT$ to~$\mathsf{co}\NP/\poly$. In this case, the runtime is subexponential. However, this contradicts the assumption on the subexponential hardness of refuting $\SAT$.
\end{proof}

\subsubsection*{AI Disclosure.}
The lower bound on the mutual information in Theorem~\ref{th:simpler-spl-in-re} was calculated by Fable 5 (Anthropic) AI agent. The encoding in the theorem is itself an immediate result of Definition~\ref{def:cute} which is generated by humans. We also used the same model to convert a hand-drawn version of Figure~\ref{fig:eps-tau} to \LaTeX{} code, and occasionally to polish a few parts of the text.

\phantom{\cite{}}
\bibliographystyle{plain}

\appendix
\section{Proof of Extended Disguising Lemma}\label{app:disguising-proof}

We first provide necessary materials for the proof.

\medskip\noindent\textbf{Two-Player games.} A two-player, simultaneous-move, zero-sum game is specified by a matrix~$\mathbf{M} \in \mathbb{R}^{a \times b}$. Player 1 chooses a row index~$i \in [a]$ and Player 2 chooses a column index~$ j \in [b]$, and Player 2 receives the payoff~$\mathbf{M}_{ij}$ from Player 1. The goal of Player 1 is minimizing the expected payoff, while Player 2 opts to maximize it. The row and column indices are called the pure strategies of Player 1 and Player 2, respectively. The mixed strategies are distributions or possible choices of indices. A mixed strategy is~$s$-uniform if it is sampled uniformly from a multiset of at most~$s$ pure strategies. 
\begin{lemma}[\cite{vNeu28}]\label{lemma:minmax}
Let~$\mathcal{P}$ and~$\mathcal{Q}$ be two mixed strategies for Player 1 and 2, respectively. It holds that
$$\min_{\mathcal{P}} \max_{j} \mathbb{E}_{i \sim \mathcal{P}} [\mathbf{M}_{ij}] = \max_{\mathcal{Q}} \min_{i} \mathbb{E}_{j \sim \mathcal{Q}} [\mathbf{M}_{ij}].$$
\end{lemma}
The value of the game, which we denote by~$\omega(\mathbf{M})$, is the optimal expected value guaranteed by the above lemma. 
The following lemma shows that each player has nearly-optimal~$s$-uniform strategy when~$s$ is chosen to be logarithm of the number of pure strategies of the opponent. 
\begin{lemma}[{\cite[Theorem~2]{LY02}}]\label{lemma:s-uniform}
For any real~$\varepsilon > 0$, any~$\vec{M} \in \mathbb{R}^{a \times b}$, and any integer~$s \geq \ln(b)/(2\varepsilon^2)$, it holds that 
\begin{align*}
\min_{\mathcal{P} \in \mathfrak{P}_s} \max_{j} \mathbb{E}_{i \sim \mathcal{P}} [\mathbf{M}_{ij}] \leq \omega(\mathbf{M}) + \varepsilon (\mathbf{M}_{\text{max}}-\mathbf{M}_{\text{min}})\ ,
\end{align*}
where~$\mathfrak{P}_s$ denotes the set of all~$s$-uniform strategies for Player 1. Similar statement holds for Player 2, namely, 
\begin{align*}
\max_{\mathcal{Q} \in \mathfrak{Q}_s} \min_{i} \mathbb{E}_{j \sim \mathcal{Q}} [\mathbf{M}_{ij}] \geq \omega(\mathbf{M}) - \varepsilon (\mathbf{M}_{\text{max}}-\mathbf{M}_{\text{min}})\ ,
\end{align*}
where~$\mathfrak{Q}_s$ denotes the set of all~$s$-uniform strategies for Player 2.
\end{lemma}

Similar to~\cite{Dru15,BBDD+20}, we define distributional stability as follows. 

\begin{definition}
Let~$n,m,m_0,m_1$ be positive integers such that~$m=m_0+m_1+1$. For a real~$\delta \in [0,1]$, a quantum mapping~$\red:\{0,1\}^{m n} \rightarrow \mathsf{MS}_{*}$ is said to be~$\delta$-quantumly-distributionally stable ($\delta$-QDS) with respect to two distributions~$(\mathcal{D}_0,\mathcal{D}_1)$ over~$\{0,1\}^n$ if the following holds:

\begin{align*}
\underset{\substack{y\sim \mathcal{D}_0\\ \pi \sim \perm_{m} }} {\mathbb{E}} \left[  \Big\|
\red \left( \pi \left(
\mathcal{D}_0^{\otimes m_0},y,\mathcal{D}_1^{\otimes m_1} \right) \right) - \red\left( \pi \left( \mathcal{D}_0^{\otimes (m_0+1)}, \mathcal{D}_1^{\otimes m_1}\right) \right) 
 \Big\| \right] \leq \delta \ .
\end{align*}
Note that the order of the pair~$(\mathcal{D}_0,\mathcal{D}_1)$ matters. 
Furthermore, when~$m_1=0$, we simply say that the mapping is~$\delta$-QDS with respect to~$\mathcal{D}_0$.
\end{definition}

Below, we recall an adaptation of~\cite[Lemma~8.10]{Dru15}.
\begin{lemma}\label{lemma:lossy1QDS}
Assume that~$\red:\{0,1\}^{m\cdot n} \rightarrow \mathsf{MS}_{*}$ satisfies the properties in Lemma~\ref{lemma:perm-dis-dis} for~$m_1=0$.
	Then~$\red$ is~$\delta$-QDS with respect to any~$ds$-uniform distribution~$\mathcal{D}_0$ supported on either~$S_0$ or~$S_1$.
\end{lemma}
In the original lemma from~\cite{Dru15}, compression is used to bound the entropy of the mutual information. However, note that this can be argued directly from splitting lossiness, and that any restriction on the input distributions will give a result for the same restricted case.

The following lemma is the generalization of the above one. 
\begin{lemma}\label{lemma:lossy2QDS}
Assume that~$\red:\{0,1\}^{m n} \rightarrow \mathsf{MS}_{*}$ satisfies the properties in Lemma~\ref{lemma:perm-dis-dis}.
	Then~$\red$ is~$\delta$-QDS with respect to any~$ds$-uniform independent distributions~$(\mathcal{D}_0,\mathcal{D}_1)$ each supported on either~$S_0$ or~$S_1$.
\end{lemma}
\begin{proof}
The proof is similar to that of~\cite[Proposition~B.1]{BBDD+20}. 
Let~$\pi \in \mathfrak{S}_m$ be a fixed permutation. One can rewrite it as the composition of two partial permutations~$\pi_0$ and~$\pi_1$, i.e.,~$\pi = \pi_0 \circ \pi_1$, such that~$\pi_1$ only acts on the last~$m_1$ arguments of the input. Let~$\rho_{\pi}(y)$ be as follows
\begin{align*}
&\rho_{\pi}(y):= \red \left( \pi \left(
\mathcal{D}_0^{\otimes m_0},y,\mathcal{D}_1^{\otimes m_1} \right) \right) \, . 
\end{align*}
For~$y,y' \sim \mathcal{D}_0$, two independent random variables, and~$\pi \sim \mathfrak{S}_m$, we want to prove that 
\begin{align*}
\mathbb{E}_{y, \pi} \left[ \Big\| 
\rho_{\pi}(y)  -  \rho_{\pi}(y')
 \Big\| \right] \leq \delta \, .
\end{align*} 
Note that it is enough to bound the conditional distributions since
\begin{align*}
\mathbb{E}_{y, \pi} \left[ \Big\| 
\rho_{\pi}(y)  -  \rho_{\pi}(y')
 \Big\| \right] =  \mathbb{E}_{\pi}\left[ \mathbb{E}_{y,\pi|\pi_1}\left[\Big\| 
\rho_{\pi}(y)  -  \rho_{\pi}(y') \Big\| \right]\right],
\end{align*}
by the law of total probability.
Let~$\red'(x_1,x_2,\ldots,x_{m_0+1})$ be the mapping that first samples $\pi$ then evaluates~$\red \left( \pi_1 \left(x_1,x_2,\ldots,x_{m_0+1},\mathcal{D}_1^{\otimes m_1} \right) \right)$.
For any fixed $\pi_1$, we show that~$\red'$ is~$\ell$-lossy for all~$ds$-uniform distributions that split over $(S_0,S_1)$.
Indeed, let~$(\mathcal{X}_1,\ldots,\mathcal{X}_{m_0+1})$ be 
independent~$ds$-uniform random variables with $\supp(\mathcal{X}_i)\subseteq S_0$ or $\supp(\mathcal{X}_i)\subseteq S_1$ for each $i\in[m_0+1]$, and $(\mathcal{Z}_1,\ldots,\mathcal{Z}_{m_1})\sim \mathcal{D}_1^{\otimes m_1}$, thus $\supp(\mathcal{Z}_i)\subseteq \supp(\mathcal D_1)\subseteq S_j$ for all $i\in[m_1]$ and some $j\in\{0,1\}$. By the splitting lossiness of~$\red$ for any~$ds$-uniform distribution, we can bound the loss of $\red'$:
\begin{align*}
\ell&\geq I_q(\pi_1(\mathcal{X}_1,\ldots,\mathcal{X}_{m_0+1},\mathcal{Z}_1,\ldots,\mathcal{Z}_{m_1});\\
&\hspace*{5.7cm}\red(\pi_1(\mathcal{X}_1,\ldots,\mathcal{X}_{m_0+1},\mathcal{Z}_1,\ldots,\mathcal{Z}_{m_1}))) \\
&= I_q(\mathcal{X}_1,\ldots,\mathcal{X}_{m_0+1},\mathcal{Z}_1,\ldots,\mathcal{Z}_{m_1};\red(\pi_1(\mathcal{X}_1,\ldots,\mathcal{X}_{m_0+1},\mathcal{Z}_1,\ldots,\mathcal{Z}_{m_1}))) \\
&\geq I_q(\mathcal{X}_1,\ldots,\mathcal{X}_{m_0+1};\red(\pi_1(\mathcal{X}_1,\ldots,\mathcal{X}_{m_0+1},\mathcal{Z}_1,\ldots,\mathcal{Z}_{m_1}))).
\end{align*}
Finally, by Lemma~\ref{lemma:lossy1QDS} a splitting lossy map must also be $\delta$-QSD, thus
\begin{align*}
&\mathbb{E}_{y, \pi| \pi_1} \left[ \Big\| 
\rho_{\pi}(y)  -  \rho_{\pi}(y')
 \Big\| \right] \\
 &= \mathbb{E}_{y,\pi| \pi_1} \left[ \Big\| 
\red \left( \pi \left(
\mathcal{D}_0^{\otimes m_0},y,\mathcal{D}_1^{\otimes m_1} \right) \right)  -  \red\left( \pi \left( \mathcal{D}_0^{\otimes m_0},y', \mathcal{D}_1^{\otimes m_1}\right) \right)\Big\| \right] \\
&= \mathbb{E}_{y,\pi_0} \left[ \Big\| 
\red' \left(
\mathcal{D}_0^{\otimes m_0},y \right)   -  \red'\left( \mathcal{D}_0^{\otimes m_0},y' \right)\Big\| \right] \\
&\leq \delta \, .
\end{align*}

\end{proof}

If a mapping is distributionally stable with respect to a pair of distributions, then one can ``sparsify'' the distributions while nearly keeping the stability. 
\begin{lemma}\label{lemma:sparsified}
Let~$n,m,m_0,m_1,\ell,S_0,S_1,\red$ and~$\delta$ be as in Lemma~\ref{lemma:perm-dis-dis}. Let~$\mathcal{D}_0$ and~$\mathcal{D}_1$ be two independent distributions with supports over $S_0$ and $S_1$, respectively. Let~$\{x^{(0)}_i\}_{i \in [d+1]}$ and~$\{x^{(1)}_i\}_{i \in [d+1]}$ be independent samples from~$\mathcal{D}_0$ and~$\mathcal{D}_1$, respectively. For each $j\in\{0,1\}$, let~$y^*_j := x^{(j)}_{i^*}$ be uniformly chosen from~$\{x^{(j)}_i\}_{i \in [d+1]}$ and let~$\widehat{\mathcal{D}}_j$ be the uniform distribution over the multiset~$\{x^{(j)}_i\}_{i \in [d+1]\setminus \{i^*\}}$. Then it holds that
\begin{align*}
\mathbb{E}_{\pi \sim \perm_m} \left[ \Big\| \red \left( \pi \left( \widehat{\mathcal{D}}_0^{\otimes m_0}, y^*_0 , \widehat{\mathcal{D}}_1^{\otimes m_1} \right) \right) - \red \left( \pi \left( \widehat{\mathcal{D}}_0^{\otimes (m_0+1)}, \widehat{\mathcal{D}}_1^{\otimes m_1} \right) \right) \Big\| \right] \\ 
\leq \delta + \frac{2m_0+1}{d+1}\, ,\\
\end{align*}
and
\begin{align*}
\mathbb{E}_{\pi \sim \perm_m} \left[ \Big\| \red \left( \pi \left( \widehat{\mathcal{D}}_0^{\otimes m_0}, y^*_1 , \widehat{\mathcal{D}}_1^{\otimes m_1} \right) \right) - \red \left( \pi \left( \widehat{\mathcal{D}}_0^{\otimes m_0}, \widehat{\mathcal{D}}_1^{\otimes (m_1+1)} \right) \right) \Big\| \right] \\ \leq \delta + \frac{2m_1+1}{d+1} \, .
\end{align*}
\end{lemma}
\begin{proof}
We prove the first statement. The other one is implied similarly.
Let~$\widetilde{\mathcal{D}}_0$ denote the uniform distribution over~$\{x^{(0)}_i\}_{i \in [d+1]}$. For any fixed set of of multisets as above and any choice of permutation~$\pi$ and quantum mapping~$\red$, we have
\begin{align*}
&\Big\| \red \left( \pi \left( \widetilde{\mathcal{D}}_0^{\otimes(m_0+1)}, \widehat{\mathcal{D}}_1^{\otimes m_1} \right) \right) - \red \left( \pi \left( \widehat{\mathcal{D}}_0^{\otimes(m_0+1)}, \widehat{\mathcal{D}}_1^{\otimes m_1} \right) \right) \Big\| \\
&\leq\Big\|  \widetilde{\mathcal{D}}_0^{\otimes(m_0+1)} \otimes \widehat{\mathcal{D}}_1^{\otimes m_1}  -  \widehat{\mathcal{D}}_0^{\otimes(m_0+1)} \otimes \widehat{\mathcal{D}}_1^{\otimes m_1}  \Big\| \\
&\leq \Big\|  \widetilde{\mathcal{D}}_0^{\otimes(m_0+1)}   -  \widehat{\mathcal{D}}_0^{\otimes(m_0+1)}   \Big\| \\
&\leq(m_0+1) \big\|  \widetilde{\mathcal{D}}_0   -  \widehat{\mathcal{D}}_0   \big\| \ ,
\end{align*}
where we used the quantum data processing inequality for the first two upper bounds, and the property of tensor product for the last one. Since both~$\widehat{\mathcal{D}}_0$ and~$\widetilde{\mathcal{D}}_0$ are classical, their trace distance coincides with their statistical distance. Therefore, we have
\begin{align*}
 \big\|  \widetilde{\mathcal{D}}_0   -  \widehat{\mathcal{D}}_0   \big\| &= \frac{1}{2} \sum\limits_{x \in \{x^{(0)}_i\}_{i \in [d+1]}} |\Pr_{\widetilde{\mathcal{D}}_0}(x) - \Pr_{\widehat{\mathcal{D}}_0}(x) | \\
&=  \frac{1}{2(d+1)} + \frac{1}{2} \sum\limits_{x \in \{x^{(0)}_i\}_{i \in [d+1]\setminus \{i^*\}}} \left| \frac{1}{d+1} - \frac{1}{d} \right| \\
&= \frac{1}{d+1} \ . 
\end{align*}
Similarly, it holds that
\begin{align*}
\Big\| \red \left( \pi \left( \widetilde{\mathcal{D}}_0^{\otimes m_0}, y^*_0 , \widehat{\mathcal{D}}_1^{\otimes m_1} \right) \right) - \red \left( \pi \left( \widehat{\mathcal{D}}_0^{\otimes m_0}, y^*_0, \widehat{\mathcal{D}}_1^{\otimes m_1} \right) \right) \Big\| \leq \frac{m_0}{d+1} \ .
\end{align*}
From the triangle inequality, it follows that
\begin{align*}
&\Big\| \red \left( \pi \left( \widehat{\mathcal{D}}_0^{\otimes m_0}, y^*_0 , \widehat{\mathcal{D}}_1^{\otimes m_1} \right) \right) - \red \left( \pi \left( \widehat{\mathcal{D}}_0^{\otimes (m_0+1)}, \widehat{\mathcal{D}}_1^{\otimes m_1} \right) \right) \Big\| \\
&\leq\Big\| \red \left( \pi \left( \widehat{\mathcal{D}}_0^{\otimes m_0}, y^*_0 , \widehat{\mathcal{D}}_1^{\otimes m_1} \right) \right) - \red \left( \pi \left( \widetilde{\mathcal{D}}_0^{\otimes m_0}, y^*_0, \widehat{\mathcal{D}}_1^{\otimes m_1} \right) \right) \Big\| \\
& \hspace{0.5cm}+ \Big\| \red \left( \pi \left( \widetilde{\mathcal{D}}_0^{\otimes m_0}, y^*_0 , \widehat{\mathcal{D}}_1^{\otimes m_1} \right) \right) - \red \left( \pi \left( \widetilde{\mathcal{D}}_0^{\otimes (m_0+1)},\widehat{\mathcal{D}}_1^{\otimes m_1} \right) \right) \Big\| \\
& \hspace{1cm} + \Big\| \red \left( \pi \left( \widetilde{\mathcal{D}}_0^{\otimes(m_0+1)}, \widehat{\mathcal{D}}_1^{\otimes m_1} \right) \right) - \red \left( \pi \left( \widehat{\mathcal{D}}_0^{\otimes(m_0+1)}, \widehat{\mathcal{D}}_1^{\otimes m_1} \right) \right) \Big\| \\
&< \Big\| \red \left( \pi \left( \widetilde{\mathcal{D}}_0^{\otimes m_0}, y^*_0 , \widehat{\mathcal{D}}_1^{\otimes m_1} \right) \right) - \red \left( \pi \left( \widetilde{\mathcal{D}}_0^{\otimes (m_0+1)},\widehat{\mathcal{D}}_1^{\otimes m_1} \right) \right) \Big\|\\
&\hspace{0.5cm}+ \frac{2m_0+1}{d+1} \ .
\end{align*}
Recall that~$\red$ is~$\ell$-lossy with respect to all~$ds$-uniform distributions that split over~$(S_0,S_1)$. Therefore, by Lemma~\ref{lemma:lossy2QDS} it is~$\delta$-QSD with respect to all~$ds$-uniform pair of distributions each supported on either~$S_0$ or~$S_1$, including~$(\widetilde{\mathcal{D}}_0, \widehat{\mathcal{D}}_1)$. Finally, by taking expectation from both sides above with respect to~$\pi$, and using the fact that~$\red$ is~$\delta$-QSD with respect to~$(\widetilde{\mathcal{D}}_0,\widehat{\mathcal{D}}_1)$, one obtains the claimed upper bound. 

\end{proof}
\color{black}

\begin{proof}[Proof of Lemma~\ref{lemma:perm-dis-dis}]

Consider the following two-player, simultaneous-move, and zero-sum game:
\begin{itemize}
\item Player 1: chooses a pair of multisets~$K \subseteq S_0$ and~$T \subseteq S_1$, each of size~$d$.
\item Player 2: chooses an element~$y \in S_0 \cup S_1$
\item Payoff: if~$y \in S_0$, Player 2 gains 
\begin{align*}
\mathbb{E}_{\pi \sim \perm_{m}} \left[ \Big\|
\red \left( \pi \left(
\unif_{K}^{\otimes m_0},y,\unif_{T}^{\otimes m_1} \right) \right) - \red\left( \pi \left( \unif_{K}^{\otimes (m_0+1)}, \unif_{T}^{\otimes m_1}\right) \right) 
 \Big\| \right] \ , 
\end{align*}
otherwise, Player 2 gains
\begin{align*}
\mathbb{E}_{\pi \sim \perm_{m}} \left[ \Big\|
\red \left( \pi \left(
\unif_{K}^{\otimes m_0},y,\unif_{T}^{\otimes m_1} \right) \right) - \red\left( \pi \left( \unif_{K}^{\otimes m_0}, \unif_{T}^{\otimes (m_1+1)}\right) \right) 
 \Big\| \right] \ .
\end{align*}
\end{itemize}

Consider a~$ds$-uniform strategy for Player 2, i.e. a distribution~$\mathcal{Y}$ of~$y$ that is uniform over a multiset of pure strategies of size~$ds$. We explain a strategy~$(\mathcal{K} , \mathcal{T})$ for Player 1 that bounds the expected payoff. 
Player 1 chooses~$K$ by sampling~$d$ independent instances of the restriction of~$\mathcal{Y}$ to~$S_0$, and chooses~$T$ by sampling~$d$ independent instances of the restriction of~$\mathcal{Y}$ to~$S_1$. The expected payoff is
\begin{multline*}
	E \eqdef \Pr_{y \sim \mathcal{Y}}(y \in S_0) \ \mathbb{E}_{\pi,K,T 
	} \left[ \Big\|
	\red \left( \pi \left(
	\unif_{K}^{\otimes m_0},y,\unif_{T}^{\otimes m_1} \right) \right) \right. \left.- \red\left( \pi \left( \unif_{K}^{\otimes (m_0+1)}, \unif_{T}^{\otimes m_1}\right) \right) 
	\Big\| \Big|  y \in S_0 \right]
\end{multline*}\vspace{-0.8cm}
\begin{multline*}
	\hspace{1cm} +  \Pr_{y \sim \mathcal{Y}}(y \in S_1) \  \mathbb{E}_{\pi,K,T 
	} \left[  \Big\|
	\red \left( \pi \left(
	\unif_{K}^{\otimes m_0},y,\unif_{T}^{\otimes m_1} \right) \right) \right. \left.- \red\left( \pi \left( \unif_{K}^{\otimes m_0}, \unif_{T}^{\otimes (m_1+1)}\right) \right) 
	\Big\|  \Big|  y \in S_1 \right] \ .
\end{multline*}
Let~$x^{(0)}_1,x^{(0)}_2,\ldots,x^{(0)}_{d+1}$ and~$x^{(1)}_1,x^{(1)}_2,\ldots,x^{(1)}_{d+1}$ be~$d+1$ independent samples from~$\mathcal{Y}\left|_{S_0} \right.$ and~$\mathcal{Y}\left|_{S_1} \right.$, respectively. Sample $i^* \rnd [d+1]$ and for~$j \in \{0,1\}$, let~$y_j^* := x^{(j)}_{i^*}$. Let~$\widehat{\mathcal{Y}}_0$ and~$\widehat{\mathcal{Y}}_1$ be the uniform distributions over the multisets~$\{x^{(0)}_i\}_{i \in [d+1]\setminus\{i^*\}}$ and~$\{x^{(1)}_i\}_{i \in [d+1]\setminus\{i^*\}}$, respectively. For~$j \in \{0,1\}$, we have that~$(y_j^*, \widehat{\mathcal{Y}}_0,\widehat{\mathcal{Y}}_1) \sim (\mathcal{Y}\left|_{S_j} \right., \mathcal{K}, \mathcal{T})$. Then, by Lemma~\ref{lemma:sparsified}, we have
\begin{multline*}
\mathbb{E}_{\pi} \left[ \Big\| \red \left( \pi \left( \widehat{\mathcal{Y}}_0^{\otimes m_0},y , \widehat{\mathcal{Y}}_1^{\otimes m_1} \right) \right) - \red \left( \pi \left( \widehat{\mathcal{Y}}_0^{\otimes (m_0+1)}, \widehat{\mathcal{Y}}_1^{\otimes m_1} \right) \right) \Big\| \ \Big| \  y \in S_0 \right] \leq \delta + \frac{2m_0+1}{d+1} \ ,
\end{multline*}
and 
\begin{multline*}
\mathbb{E}_{\pi} \left[ \Big\| \red \left( \pi \left( \widehat{\mathcal{Y}}_0^{\otimes m_0}, y , \widehat{\mathcal{Y}}_1^{\otimes m_1} \right) \right) - \red \left( \pi \left( \widehat{\mathcal{Y}}_0^{\otimes m_0}, \widehat{\mathcal{Y}}_1^{\otimes (m_1+1)} \right) \right) \Big\| \ \Big| \  y \in S_1 \right] \leq \delta + \frac{2m_1+1}{d+1} \ .
\end{multline*}
Therefore, we obtain~$E \leq \delta + 2(m+1)/(d+1)$. 

Above, we showed that for every~$ds$-uniform strategy for Player 2, there exists a strategy for Player 1 that bounds the expected payoff by~$\delta + 2(m+1)/(d+1)$. Let~$\mathbf{M}:=[\mathbf{M}_{ij}]_{i,j}$ be the matrix such that~$\mathbf{M}_{ij}$ corresponds to the payoff when Player 1 outputs~$i$ and Player 2 outputs~$j$. By Lemma~\ref{lemma:s-uniform}, we have
\begin{multline*}
\delta + 2(m+1)/(d+1) \geq \max_{\mathcal{Q} \in \mathfrak{Q}_{ds}} \min_{i} \mathbb{E}_{j \sim \mathcal{Q}} [\mathbf{M}_{ij}] \geq \omega(\mathbf{M}) - \varepsilon (\mathbf{M}_{\text{max}}-\mathbf{M}_{\text{min}}) \geq \omega(\mathbf{M}) - \varepsilon \ ,
\end{multline*}
where~$\mathfrak{Q}_{ds}$ is the set of all~$ds$-uniform strategies for Player 2. It follows that $\omega(\mathbf{M}) \leq 	\delta + 2(m+1)/(d+1) + \varepsilon$.

Now we use Lemma~\ref{lemma:s-uniform} in other way around. In fact, the number of possible choices for Player~1 is~$|S_0 \cup S_1| \leq 2^n$. Therefore, Lemma~\ref{lemma:s-uniform} asserts that there exists a~$s$-uniform strategy for Player 2 such that for any possibly mixed strategy for Player 1, the expected payoff is at most~$\varepsilon$-far from the value of the game~$\omega(\mathbf{M})$. In other words, for this particular strategy of Player 1, the expected payoff is always at most
$$\omega(\mathbf{M}) + \varepsilon \leq \delta + 2(m+1)/(d+1) + 2\varepsilon \, .$$
Recall that a $s$-uniform strategy is, by definition, a uniformly sampled element from a size-$s$ multiset of choices of the player. Note that Player~1 chooses a pair~$(K,T)$. Therefore, this strategy is essentially a uniform distribution over some multiset~$\{(K_1,T_1),\ldots,(K_s,T_s)\}$, which concludes the proof.

\end{proof}

\newpage
\tableofcontents

\end{document}